\newif\ifarxiv
\arxivtrue 

\ifarxiv
  \documentclass[11pt]{article}
  \usepackage[margin=1in]{geometry}
  \usepackage[slantedGreek]{mathpazo}
  \usepackage{graphicx}
  \usepackage{amsfonts}
  \usepackage{amssymb}
  \usepackage{amsthm}
  \usepackage{float}
  \usepackage{amsmath}%
  \usepackage[square,sort,comma]{natbib}
  \usepackage[colorlinks,citecolor=magenta,urlcolor=blue]{hyperref} 
\else
\fi

\usepackage{color}

\usepackage{tikz}
\usetikzlibrary{positioning}

\usepackage{subcaption}

\renewcommand{\det}[1]{\left| #1 \right|}

\newcommand{\ben}{\begin{enumerate}}
\newcommand{\een}{\end{enumerate}}
\newcommand{\beq}{\begin{equation}}
\newcommand{\eeq}{\end{equation}}
\newcommand{\bde}{\begin{description}}
\newcommand{\ede}{\end{description}}

\newcommand{\NormRV}{\mathcal{N}}

\newcommand{\CauchyRV}{\mathcal{C}a}

\newtheoremstyle{slplain}
  {1\baselineskip\@plus.2\baselineskip\@minus.2\baselineskip}
  {.5\baselineskip\@plus.2\baselineskip\@minus.2\baselineskip}
  {\slshape}
  {}
  {\bfseries}
  {.}
  { }
  {}

\theoremstyle{slplain}

\newtheorem{theorem}{Theorem}

\newtheorem{corollary}[theorem]{Corollary}

\newtheorem{definition}[theorem]{Definition}

\newtheorem{lemma}[theorem]{Lemma}

\newtheorem{proposition}[theorem]{Proposition}
\newtheorem{remark}[theorem]{Remark}

\usepackage{booktabs,array}

\newcount\rowc

\newtheorem{assumption}{Assumption}

\ifarxiv
\title{Bayesian Global-Local Regularization}
\author{	
	\makebox[.4\linewidth]{Jyotishka Datta}\\\textit{Department of Statistics}\\\textit{Virginia Tech}\\\and 
	\makebox[.4\linewidth]{Nick Polson}\\\textit{ Chicago Booth}\\\textit{  University of Chicago}\\\and 
	\makebox[.4\linewidth]{Vadim Sokolov\footnote{Jyotishka Datta is Assistant Professor in the Department of Statistics at Virginia Tech. Nick Polson is Professor of Econometrics and Statistics at Chicago Booth, University of Chicago: ngp@chicagobooth.edu. Vadim Sokolov is Associate Professor at Volgenau School of Engineering, George Mason University, USA: vsokolov@gmu.edu.}}\\\textit{ Department of Systems Engineering }\\\textit{  and Operations Research}\\\textit{ George Mason University}
}
\date{First Draft: January 1, 2025 \\This Draft: \today}
\else
\fi

\graphicspath{{./fig/}}

\begin{document}
\ifarxiv
\maketitle
\begin{abstract}
\noindent We propose a unified framework for global-local regularization that bridges the gap between classical techniques---such as ridge regression and the nonnegative garotte---and modern Bayesian hierarchical modeling. By estimating local regularization strengths via marginal likelihood under order constraints, our approach generalizes Stein's positive-part estimator and provides a principled mechanism for adaptive shrinkage in high-dimensional settings. We establish that this isotonic empirical Bayes estimator achieves near-minimax risk (up to logarithmic factors) over sparse ordered model classes, constituting a significant advance in high-dimensional statistical inference. Applications to orthogonal polynomial regression demonstrate the methodology's flexibility, while our theoretical results clarify the connections between empirical Bayes, shape-constrained estimation, and degrees-of-freedom adjustments.
\end{abstract}
\else
\fi

\noindent Keywords: Bayesian, hyperparameter choice, model selection, regularization. 

\newpage

\section{Introduction}

Regularization lies at the heart of modern high-dimensional statistics and machine learning, serving as the primary mechanism to control the bias-variance trade-off \citep{hastie2016elements}. In regimes where the number of parameters $p$ rivals or exceeds the sample size $n$, classical unregularized estimators fail, and one must impose structure---sparsity, smoothness, or shape constraints---to recover the signal. Historically, two distinct cultures have addressed this challenge: the classical frequentist approach, exemplified by ridge regression \citep{hoerl1970ridge}, subset selection, and the nonnegative garotte \citep{breiman1995better}; and the Bayesian hierarchical modeling approach, which relies on global-local shrinkage priors \citep{polson2010shrink, polson2012local} to adaptively penalize coefficients.

While both traditions share the goal of complexity reduction, they often operate in parallel with distinct theoretical languages and algorithmic tools. Classical methods frequently rely on cross-validation or $C_p$-type criteria for tuning, which can be computationally intensive and unstable \citep{shao1993linear}. Conversely, fully Bayesian approaches offer principled uncertainty quantification but can suffer from computational bottlenecks and sensitivity to hyperparameter specifications in high dimensions.

In this paper, we propose a significant advance in high-dimensional inference: a \textit{unified global-local regularization framework} that bridges these two worlds. Our methodology interprets the classical regularization penalties as arising from hierarchical Gaussian priors where the local variances are treated as hyperparameters to be estimated. A key innovation of our work is the use of \textit{isotonic regression} to estimate these local variance profiles under natural ordering constraints. This strategy, inspired by the empirical Bayes insights of \citet{deaton1980empirical} and \citet{xu2007some}, allows for a data-driven discovery of the signal structure---whether it be the decay of coefficients in a polynomial expansion or the sparsity pattern in a wavelet basis---without the need for ad-hoc thresholds or discrete model selection.

Our contributions are threefold. First, we formalize the connection between generalized ridge regression and isotonic empirical Bayes estimation, showing how methods like the nonnegative garotte and Stein's positive-part estimator \citep{james1961estimation} emerge as special cases of our framework. Second, we provide rigorous theoretical guarantees for our method. We prove that our isotonic empirical Bayes estimator achieves near-minimax optimal risk (up to logarithmic factors) over a broad class of sparse, ordered signals (Theorem~\ref{thm:main}). This result serves as a theoretical anchor, validating that the flexibility of our global-local formulation does not come at the cost of statistical efficiency. Finally, we demonstrate the practical utility of our framework through applications to orthogonal polynomial regression, illustrating how ``ordering'' priors can effectively regularize complex models.

For ease of reading, we summarize the contributions more explicitly:
\begin{enumerate}
\item We introduce an ordered-variance empirical Bayes formulation of global--local shrinkage in which the local variance profile is estimated by marginal likelihood under monotonicity constraints, yielding a simple isotonic regression (PAVA) computation.
\item We establish adaptive, near-minimax risk bounds (up to logarithmic factors) for the resulting shrinkage rule over ordered sparse classes, and we provide an end-to-end (unconditional) guarantee under an explicit dyadic-margin condition.
\item We connect the resulting shrinkage profiles to classical regularization rules in canonical coordinates, clarifying how ridge-type procedures and Stein-type shrinkage can be viewed within the same variance-component framework.
\end{enumerate}

The remainder of the paper is organized as follows. Section~\ref{sec:background} introduces the global--local framework and motivating examples (including Deaton/Xu-style ordered shrinkage). Section~\ref{sec:global_local_framework} explains the regression/regularization connections in canonical coordinates. Section~\ref{sec:risk_bounds} states our main frequentist risk guarantee for isotonic empirical Bayes shrinkage and clarifies the role of Gaussian cloning/cross-fitting as a proof device. Section~\ref{sec:concentration_results} then gives posterior contraction statements for the isotonic EB posterior.

The rest of the paper is organized as follows.  Section~\ref{sec:background} by establishes the foundational global--local shrinkage framework in the context of Gaussian sequence models. Then we demonstrate when empirical Bayes estimation collapses and when it remains stable, we review classical perspectives on variance component estimation including ANOVA priors, the connections to Tikhonov regularization, and Stein's positive-part rule. These historical approaches provide conceptual context for our isotonic empirical Bayes methodology. The section concludes with motivating examples from Deaton and Xu that illustrate how ordered-variance structures arise naturally in polynomial regression and spectral decompositions.

Section~\ref{sec:global_local_framework} connects the ordered-variance empirical Bayes perspective to classical regularization in regression. We show how many familiar regularization procedures---ridge regression, generalized ridge regression, principal components regression, and positive-part Stein shrinkage---emerge as special cases of hierarchical Gaussian models with structured variance components. The key insight is that regularization rules become coordinatewise shrinkage rules in canonical (orthogonal) coordinates, and these shrinkage factors correspond to variance components that can be estimated under monotonicity constraints.

Section~\ref{sec:risk_bounds} provides the main theoretical contribution: frequentist risk bounds for isotonic empirical Bayes global--local shrinkage. We prove that the isotonic empirical Bayes estimator achieves near-minimax optimal risk (up to logarithmic factors) over ordered sparse classes. The analysis employs a Gaussian cloning technique that decouples hyperparameter estimation from posterior mean computation, enabling sharp risk characterization. We also provide an end-to-end unconditional guarantee under an explicit dyadic-margin condition that quantifies when the signal structure is sufficiently separated from noise.

Section~\ref{sec:concentration_results} translates the frequentist risk bounds into Bayesian posterior contraction results. We establish that the Gaussian posterior induced by isotonic empirical Bayes shrinkage contracts around the true parameter at near-minimax rates in $\ell_2$ norm. These concentration guarantees complement the risk bounds and demonstrate that the empirical Bayes posterior provides valid uncertainty quantification.

Finally, Section~\ref{sec:discussion_conclusion} synthesizes the main insights and discusses methodological implications, computational considerations, and directions for future work. The appendices contain detailed proofs, extensions to unknown noise variance, connections to Sobolev-type smoothness classes, and variants of the main results.

\section{Global and Local Shrinkage}
\label{sec:background}

Many statistical estimation problems can be recast as \textit{Gaussian sequence models}. In linear regression, the singular value decomposition $X = U D V^\top$ allows projection onto the canonical basis: $\hat{\alpha}_i \sim \mathcal{N}(\alpha_i, \sigma^2/d_i^2)$, where $\alpha_i = V^\top \beta$ are rotated coefficients. Similarly, in polynomial regression with orthogonal basis functions (e.g. Legendre polynomials), the regression function $f(x) = \sum_j \theta_j \psi_j(x)$ leads to coefficient estimates $\hat{\theta}_j \sim \mathcal{N}(\theta_j, \sigma^2/n)$. Both models reduce to heteroskedastic (or homoskedastic) Gaussian means problems, enabling shrinkage estimation via global-local priors. This perspective unifies classical regularization and Bayesian shrinkage in a common framework.

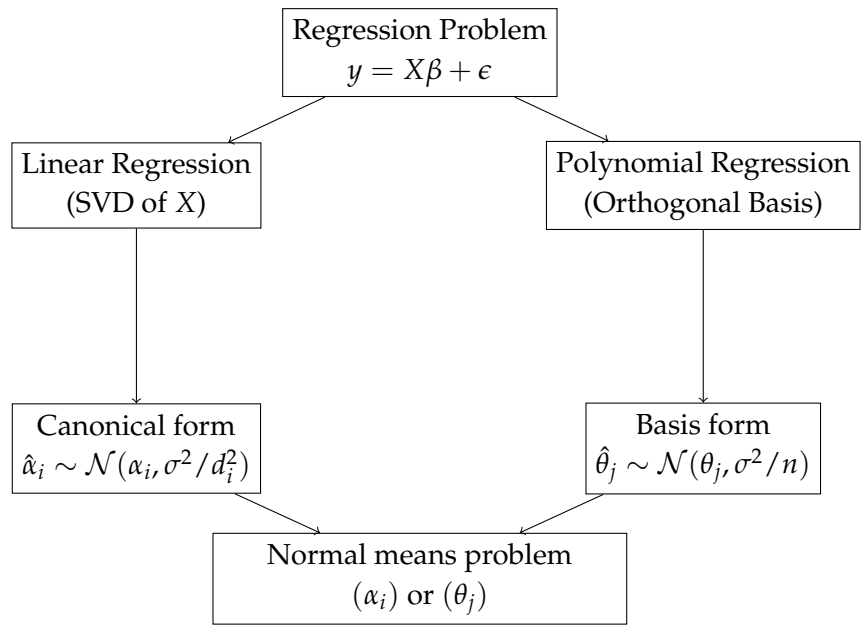
\begin{figure}[H]
\centering
\begin{tikzpicture}[node distance=2.5cm, every node/.style={align=center}]
  \node (regression) [rectangle, draw, minimum width=3.5cm, minimum height=1cm] {Regression Problem \\ $y = X\beta + \epsilon$};
  \node (svd) [rectangle, draw, below left of=regression, xshift=-2cm] {Linear Regression \\ (SVD of $X$)};
  \node (poly) [rectangle, draw, below right of=regression, xshift=2cm] {Polynomial Regression \\ (Orthogonal Basis)};
  \node (means1) [rectangle, draw, below of=svd, yshift=-1cm] {Canonical form \\ $\hat{\alpha}_i \sim \mathcal{N}(\alpha_i, \sigma^2/d_i^2)$};
  \node (means2) [rectangle, draw, below of=poly, yshift=-1cm] {Basis form \\ $\hat{\theta}_j \sim \mathcal{N}(\theta_j, \sigma^2/n)$};
  \node (gl) [rectangle, draw, below of=regression, yshift=-4.5cm, minimum width=5.5cm] {Normal means problem \\ $(\alpha_i)$ or $(\theta_j)$};

  \draw[->] (regression) -- (svd);
  \draw[->] (regression) -- (poly);
  \draw[->] (svd) -- (means1);
  \draw[->] (poly) -- (means2);
  \draw[->] (means1) -- (gl);
  \draw[->] (means2) -- (gl);
\end{tikzpicture}
\caption{Both linear and polynomial regression reduce to normal means problems under suitable transformations.}
\label{fig:normal_means_reduction}
\end{figure}

The popular global-local (GL, henceforth) regression framework typically involves a set of local shrinkage parameters ($\lambda_i^2$) that identify the non-null effects and two global shrinkage parameters: one for controlling the amount of sparsity ($\tau$) and another global variance parameter $\sigma^2$. For example, in the context of a Gaussian sequence model or a sparse normal means model, we define a standard GL framework as $Y_i \mid \theta_i \sim \NormRV(\theta_i, \sigma^2)$, with $\theta_i \sim \NormRV(0, \lambda_i^2 \tau^2 \sigma^2)$, $\lambda_i \sim p(\lambda_i)$ and $\sigma^2, \tau^2 \sim p(\sigma^2, \tau^2)$, with the horseshoe prior \citep{carvalho2010horseshoe} corresponding to the half-Cauchy prior on the local $\lambda_i$. Handling the global shrinkage parameter $\tau$ can be a delicate issue, as the marginal likelihood has a mode at $0$ \citep{tiao1965bayesian}, and therefore an empirical Bayes approach has a potential, possibly rare, danger of collapsing to zero \citep{carvalho2009handling}. For both the full Bayes and empirical Bayes approach, multiple authors have proved the optimality of the posterior mean, both in terms of minimaxicity and decision-theoretic Bayes risk \citep{datta2013asymptotic, van2014asymptotically, van2014horseshoe, van2015conditions, ghosh2016asymptotic, ghosh2016testing}. \cite{piironen2017sparsity}, on  the other hand, provides an effective model size approach in which the connections between the sparsity level and the global shrinkage parameter $\tau$ have been exploited to handle $\tau$. The collapse of $\tau$ to zero, especially in small-sample settings or under absolute scaling, has been studied in both computational and theoretical work. \citet{carvalho2009handling} and \citet{piironen2017sparsity} noted that empirical Bayes estimation can suffer from this degeneracy, and \citet{van2016many} attempted to mitigate it via truncation strategies. 

\subsection{Global Parameters: Unknown \texorpdfstring{$\sigma^2, \tau^2$}{sigma2, tau2}} 

Although handling the global shrinkage parameter has received a lot of attention in the GL literature, the issue of the global variance component as it relates to $\tau$ is relatively less explored, and theoretical work often proceeds by assuming $\sigma^2 = 1$. However, 
\citet{polson2010shrink} emphasized this issue, demonstrating that absolute priors such as $\tau \sim \CauchyRV^+(0, 1)$ can lead to posterior bimodality, while $\tau \sim \CauchyRV^+(0, \sigma)$ stabilizes the inference. Their example, with only two observations and a true signal $\theta_i = 20$, showed drastic differences between the two priors. In our simulations, we replicate their toy example by comparing the horseshoe posterior under four combinations of $\tau$ and $\sigma^2$ priors. In this experiment, either both parameters are estimated via an empirical Bayes (MMLE) approach, or $\sigma^2$ is assigned a Jeffreys' prior, and for $\tau$, we use either a $\CauchyRV^+(0,1)$ prior, a $\CauchyRV^+(0,\sigma)$ prior, or an empirical Bayes approach. We find that two out of the four scenarios avoid collapse and produce coherent unimodal posteriors: the one where both $\sigma^2$ and $\tau^2$ are fixed or $p(\sigma^2) \propto 1/\sigma^2$ and $\CauchyRV^+(0,\sigma)$. Similar phenomena are reported by \citet{lee2024tail}, who propose a tail-adaptive prior to handle collapsing behavior, and in empirical work by \citet{abba2018adapting} who proposed a Bayesian analog of the square root lasso method \citet{belloni2011square}. 

\begin{figure}[ht!]
    \centering
    \includegraphics[width=0.7\textwidth]{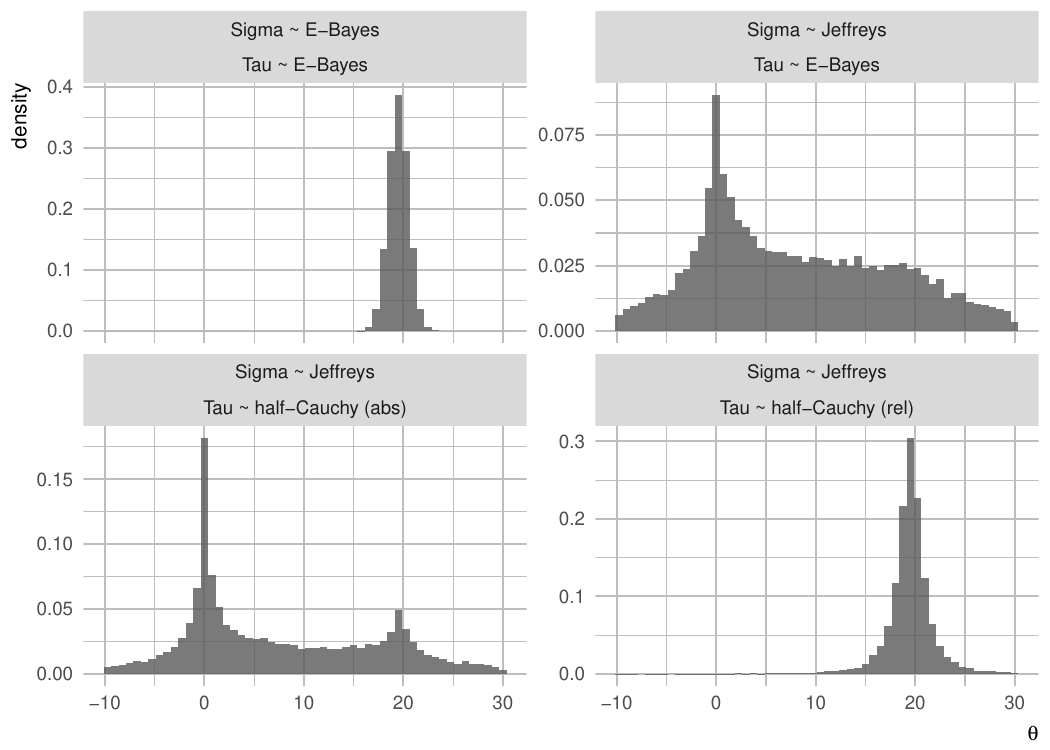}
    \caption{Effect of different choices of global and noise prior scaling on posterior shrinkage profiles in Horseshoe regression. Relative scaling of $\tau$ with $\sigma$ avoids the collapsing behavior seen in absolute scaling.}
    \label{fig:effect_of_sigma}
\end{figure}

\subsection{Collapse properties in Gaussian means}
\label{sec:collapse_properties}

Empirical Bayes (Type-II maximum likelihood) selects hyperparameters by maximizing the marginal likelihood. In many shrinkage models this optimization can concentrate on the boundary and set a scale parameter to zero, producing a fully collapsed fit (all coefficients shrunk to the prior mean). Because the marginal likelihood is typically intractable in richer models, it is helpful to start in a setting where everything is explicit: we can characterize exactly when the empirical Bayes estimate collapses, and when it stays away from zero.

The Gaussian means model is the simplest case where the EB objective is analytic. It yields a necessary and sufficient condition for non-collapse, \(\hat\tau^2>0\), and it serves as a template for global--local extensions in more structured problems, where the ambient dimension \(p\) is naturally replaced by an ``effective dimension'' determined by the design or spectrum.

In particular, the characterization below shows that EB collapse, \(\hat\tau^2=0\), occurs exactly when the observed energy \(\|Y\|_2^2\) does not exceed the noise-level benchmark \(p\sigma^2\).

Consider the Gaussian sequence model
\begin{equation}
Y_i = \theta_i + \varepsilon_i,
\qquad
\varepsilon_i \stackrel{\mathrm{i.i.d.}}{\sim} N(0,\sigma^2),
\quad i=1,\dots,p,
\label{eq:means_sigma}
\end{equation}
with known $\sigma^2>0$. Suppose a global Gaussian prior
\begin{equation}
\theta_i \mid \tau^2 \sim N(0,\tau^2\sigma^2),
\qquad i=1,\dots,p,
\label{eq:global_gauss_prior}
\end{equation}
where $\tau^2\ge 0$ is the global shrinkage/variance parameter.
Then the marginal model is
\begin{equation}
Y_i \mid \tau^2 \sim N\!\big(0,\sigma^2(1+\tau^2)\big),
\quad \text{independent across } i.
\label{eq:marg}
\end{equation}
Let $\hat\tau^2$ denote the Type-II ML estimator (empirical Bayes):
\[
\hat\tau^2 \in \arg\max_{\tau^2\ge 0}\ \log m(Y\mid \tau^2).
\]

\begin{theorem}[Exact characterization of EB collapse/non-collapse]
\label{thm:exact_collapse}
In model \eqref{eq:means_sigma}--\eqref{eq:global_gauss_prior}, the empirical Bayes estimator has the closed form
\begin{equation}
\hat\tau^2 \;=\;\left(\frac{\|Y\|_2^2}{p\sigma^2}-1\right)_+.
\label{eq:tau_hat_closed}
\end{equation}
Equivalently,
\begin{equation}
\hat\tau^2 = 0
\quad\Longleftrightarrow\quad
\|Y\|_2^2 \le p\sigma^2,
\qquad
\hat\tau^2>0
\quad\Longleftrightarrow\quad
\|Y\|_2^2 > p\sigma^2.
\label{eq:collapse_iff}
\end{equation}
\end{theorem}

\begin{proof}
From \eqref{eq:marg},
\[
\log m(Y\mid \tau^2)
:=
-\frac{p}{2}\log\!\big(2\pi\sigma^2(1+\tau^2)\big)
-\frac{1}{2\sigma^2(1+\tau^2)}\|Y\|_2^2.
\]
Differentiate with respect to $\tau^2$:
\[
\frac{\partial}{\partial \tau^2}\log m(Y\mid\tau^2)
:=
-\frac{p}{2}\frac{1}{1+\tau^2}
+\frac{1}{2\sigma^2}\frac{\|Y\|_2^2}{(1+\tau^2)^2}.
\]
Setting the derivative to zero yields
\[
\frac{\|Y\|_2^2}{\sigma^2(1+\tau^2)} = p
\quad\Longleftrightarrow\quad
\tau^2 = \frac{\|Y\|_2^2}{p\sigma^2}-1.
\]
Since $\tau^2\ge 0$, the constrained maximizer is \eqref{eq:tau_hat_closed}.
The equivalences in \eqref{eq:collapse_iff} follow immediately.
\end{proof}

\subsection{ANOVA (Priors for Variance Components)}  

Having discussed the role of global variance components such as $\sigma^2$ and $\tau^2$, we now focus on the specification and estimation of variance components across model parameters in hierarchical and ANOVA-type settings. In the Bayesian one-way random effects model, the observations are modeled as  
\begin{equation}
y_{ij} = \mu + a_i + e_{ij}, \quad i = 1, \dots, k; \quad j = 1, \dots, n,
\label{eq:random_effects_model}
\end{equation}
where $\mu$ is a common location parameter, $a_i \sim N(0, \sigma_a^2)$ represents the random effect associated with the group $i$, and $e_{ij} \sim N(0, \sigma^2)$ is the residual error. The $a_i$ and $e_{ij}$ are assumed to be independent. In this model, the variance of an individual observation is $\operatorname{Var}(y_{ij}) = \sigma^2 + \sigma_a^2$, and the likelihood depends on both variance components.

The joint likelihood of $(\mu, \sigma^2, \sigma_a^2)$ is given by
\begin{equation}
\begin{aligned}
L(\mu, \sigma^2, \sigma_a^2 \mid \mathbf{y}) &\propto (\sigma^2)^{-k(n-1)/2} (\sigma^2 + n\sigma_a^2)^{-k/2} \\
&\quad \times \exp\left\{ -\frac{1}{2} \left[ \frac{S_1}{\sigma^2} + \frac{S_2}{\sigma^2 + n\sigma_a^2} + \frac{nk(\bar{y} - \mu)^2}{\sigma^2 + n\sigma_a^2} \right] \right\},
\end{aligned}
\label{eq:likelihood}
\end{equation}
where $S_1 = \sum_{i=1}^k \sum_{j=1}^n (y_{ij} - \bar{y}_{i\cdot})^2$ is the sum of squares within the group and $S_2 = \sum_{i=1}^k n(\bar{y}_{i\cdot} - \bar{y}_{\cdot\cdot})^2$ is the sum of squares between-groups.

It is well known that classical estimators, such as the ANOVA-based estimator of $\sigma_a^2$, can take negative values even when variances are constrained to be nonnegative. Hill \citep{hill1965inference} documented this issue through examples in which the unbiased estimator for $\sigma_a^2$ is negative with positive probability. To address this, Bayesian methods impose proper or conditionally proper priors on the variance components. \citet{chaloner1987bayesian} systematically compared several prior distributions for $\sigma^2$ and $\sigma_a^2$, analyzing their impact on posterior propriety and inference.

In particular, she considered the following three noninformative joint priors:
\begin{equation}
p(\sigma_a^2, \sigma^2) \propto (\sigma_a^2 + \sigma^2)^{-2},
\label{eq:prior1}
\end{equation}
\begin{equation}
p(\sigma_a^2, \sigma^2) \propto (\sigma^2)^{-1}(\sigma_a^2 + \sigma^2)^{-1},
\label{eq:prior2}
\end{equation}
\begin{equation}
p(\sigma_a^2, \sigma^2) \propto (\sigma^2)^{-1/2}(\sigma_a^2 + \sigma^2)^{-3/2},
\label{eq:prior3}
\end{equation}
each motivated by different invariance considerations. The first of these (Equation~\ref{eq:prior1}) corresponds to a uniform prior on the intraclass correlation coefficient $\rho = \sigma_a^2 / (\sigma_a^2 + \sigma^2)$. Chaloner found that all three priors led to similar posterior inferences in balanced designs, and used Equation~\ref{eq:prior1} for simulations due to its interpretability.

In addition to these noninformative priors, Chaloner also evaluated informative priors for $\sigma_a^2$ based on inverse gamma distributions. Specifically, she used inverse exponential distributions by setting the shape parameter $\alpha = 1$ and varying the scale $\beta$, so that the prior takes the form
\begin{equation}
p(\mu, \sigma^2, \sigma_a^2) \propto (\sigma^2)^{-1} (\sigma_a^2)^{-2} \exp\left\{ -\frac{1}{\beta \sigma_a^2} \right\},
\label{eq:invexp}
\end{equation}
where the prior mode of $\sigma_a^2$ is $\beta^{-1} / 2$. Adjusting $\beta$, she evaluated priors centered on 0.1, 1.0 and 10, showing that posterior summaries were robust to moderate prior changes.

Finally, these structures are connected with modern global-local shrinkage priors \citep{polson2010shrink, polson2012local,polson2012half}. A commonly used independent Jeffreys-like prior is
\begin{equation}
p(\sigma^2, \sigma_a^2) = \frac{1}{\sigma^2} \cdot \frac{1}{\sigma_a^2},
\label{eq:jeffreys}
\end{equation}
while alternative joint priors such as
\begin{equation}
p(\sigma^2, \sigma_a^2) \propto \frac{1}{\sigma^2} \cdot \frac{1}{\sigma^2 + \sigma_a^2}
\label{eq:dependent}
\end{equation}
penalize imbalance between variance components. A critical insight from \citet{polson2012half} is that while Jeffreys' prior \ref{eq:jeffreys} on the error variance $p(\sigma^2) \propto \sigma^{-2}$ generally poses no problem, applying the same form to local scale parameters—e.g., $p(\lambda^2) \propto \lambda^{-2}$—leads to an improper posterior. This motivates the use of proper shrinkage priors like the half-Cauchy, which are integrable at the origin and yield well-behaved posteriors even under heavy shrinkage.

Equation~\ref{eq:dependent} is closely related to global-local shrinkage priors such as the half-Cauchy prior $C^+(0, \sigma)$, advocated by \citet{gelman2006prior} and extended by \citet{carvalho2010horseshoe, polson2010shrink, polson2012half} in hierarchical and high-dimensional inference. Recent work has revisited the foundational insights of \citet{tiao1965bayesian}, particularly in the context of adaptive shrinkage. \citet{polson2012half} note that the decomposition of variance in hierarchical models by Tiao and Tan directly anticipates modern scale-mixture priors, enabling flexible and computationally tractable shrinkage rules. In a related development, \citet{morris2011hierarchical} examined the one-way random effects model from an empirical Bayes perspective, connecting the classical ANOVA decomposition to contemporary marginal likelihood approaches to estimate variance components and random effects in high-dimensional settings.

\subsection{Example: Tikhonov Regularization}
In a foundational paper, \citet{mackay1992bayesian} showed how regularized regression problems such as ridge and Tikhonov estimation can be interpreted as Bayesian inference under Gaussian priors, where the regularization parameters correspond to global scale parameters in the prior. In MacKay’s framework, these global scales modulate the prior precision matrix over regression coefficients and appear as hyperparameters in marginal likelihood, which can be estimated using empirical Bayes or integrated with a fully Bayesian treatment.

The Tikhonov regularization framework provides a general setting for regression. Given observed data $\{(x_i, y_i)\}_{i=1}^n$ and a parametric model $f(x,w)$ with parameter vector $w \in \mathbb{R}^k$, we often define a data misfit functional
\[
E_D(w) = \sum_{i=1}^n \left( y_i - f(x_i,w) \right)^2,
\]
which yields a Gaussian likelihood
\[
p(y \mid w) = \frac{1}{Z_D}\,\exp\!\left(-\frac{1}{2\sigma^2}\,E_D(w)\right),
\]
where $\sigma^2$ denotes the noise variance and $Z_D$ is a normalization constant.

A Gaussian prior on the weights can be specified as
\[
p(w) = \frac{1}{Z_W}\,\exp\!\left(-\frac{1}{2\sigma_w^2}\,E_W(w)\right),
\]
where $E_W(w)$ denotes a quadratic penalty on $w$ (e.g. $E_W(w)=w^T w$), $\sigma_w^2$ is the prior variance, and $Z_W$ is its normalization constant. See \citet{BoxTiao1964} for a discussion of such priors. The hyperparameters $\sigma^2$ and $\sigma_w^2$ control the strength of the noise and the prior, and are commonly denoted by $\beta^{-1}=\sigma^2$ and $\alpha^{-1}=\sigma_w^2$ in the notation of MacKay~\citep{mackay1992bayesian}.

The normalization constants $Z_D$ and $Z_W$ can be viewed as partition functions\footnote{A \emph{partition function} is the integral (or sum) of an unnormalized exponential density, serving as the normalization constant that makes the distribution proper, and it plays a role in computing marginal likelihoods in Bayesian inference and has parallels in statistical mechanics.} for the data and prior distributions, respectively, and together determine the marginal likelihood (or model evidence). The product of $Z_D$ and $Z_W$ is related to the marginal likelihood (evidence) $p(D \mid \sigma^2, \sigma_w^2)$. Estimation of these hyperparameters can be performed by Type-II maximum likelihood or by a full Bayes treatment. 

Following MacKay~\citep{mackay1992bayesian}, define the precisions $\tau_D^2 = 1/\sigma^2$ and $\tau_w^2 = 1/\sigma_w^2$, and let $B$ and $C$ denote the Hessians of $E_D(w)$ and $E_W(w)$ in the maximum a posteriori (MAP) estimate $w^{\mathrm{MP}}$. The Hessian of the negative log-posterior is then $A = \tau_w^2 C + \tau_D^2 B$. Evaluating the log-evidence under a Gaussian Laplace approximation yields
\[
\log p(D \mid \tau_w^2,\tau_D^2) = - \tau_w^2 E_w^{\mathrm{MP}} - \tau_D^2 E_D^{\mathrm{MP}} - \frac{1}{2}\log\det A - \log Z_W(\tau_w^2) - \log Z_D(\tau_D^2) + \frac{k}{2}\log(2\pi),
\]
where $E_D^{\mathrm{MP}}$ and $E_W^{\mathrm{MP}}$ denote the data and prior energy terms at the MAP estimate, and $k$ is the number of parameters. Equivalently, the evidence can be written as
\[
p(D \mid \tau_w^2,\tau_D^2) = \frac{Z_M(\tau_w^2,\tau_D^2)}{Z_W(\tau_w^2)\, Z_D(\tau_D^2)},
\]
where $Z_M$ is the normalization of the posterior.

MacKay~\citep{mackay1992bayesian} argued that this evidence is preferable to cross-validation. The associated Occam factor, which penalizes excessively small prior variances, is
\[
-\tau_w^2 E_W^{\mathrm{MP}} - \frac{1}{2}\log\det A - \log Z_W(\tau_w^2),
\]
and can also be expressed as
\[
(2\pi)^{k/2}\det{A}^{-1/2}/Z_W(\tau_w^2),
\]
Represents the reduction in the effective volume of parameter space from before to after.

Setting the derivatives of the logarithmic evidence with respect to $\tau_w^2$ to zero yields the most probable value $\hat{\tau}_w^2$, which satisfies
\[
2 \hat{\tau}_w^2 E_W^{\mathrm{MP}} = k - \hat{\tau}_w^2\,\mathrm{tr}(A^{-1}).
\]
Defining the prior penalty misfit as $\chi_w^2 = 2\tau_w^2 E_W(w)$ and the effective number of parameters as
\[
\gamma = k - \hat{\tau}_w^2\,\mathrm{tr}(A^{-1}) = k - \sum_{a=1}^k \frac{\hat{\tau}_w^2}{\lambda_a + \hat{\tau}_w^2},
\]
where $\{\lambda_a\}_{a=1}^k$ are the eigenvalues of $\tau_D^2 B$, we recover the usual degree-of-freedom adjustment. Similarly, the data mismatch is measured by $\chi_D^2 = 2\tau_D^2 E_D(w)$.

\subsection{Example: Stein's Positive-Part Estimator as a Bayes Rule}
A classic illustration of the interplay between shrinkage and Bayesian inference is Stein’s positive-part estimator, which arises as the posterior mode under a particular prior on the shrinkage weight; this perspective, originally developed by \citet{takada1979steins}, reframes Stein’s estimator as a decision-theoretic Bayes rule with a data-driven shrinkage coefficient.


Consider the classical problem of estimating a $p$-dimensional normal mean:
\[
y_i \mid \theta_i \sim N(\theta_i, 1), \quad \text{for } i = 1, \dots, p.
\]
with $\theta_i \sim N(0, \tau^2)$. The posterior mean under this prior is
\[
\mathbb{E}[\theta_i \mid y_i, \tau^2] = \left( 1 - \kappa \right) y_i, \quad \text{where } \kappa = \frac{1}{1 + \tau^2}.
\]
Reparameterizing the shrinkage as a function of $\kappa$, Takada \citep{takada1979steins} proposed placing a prior directly on $\kappa$. Specifically, let
\[
p(\kappa) \propto (1 - \kappa)^{p/2} \kappa^{-pa/2}, \quad \text{where } a = 1 - \frac{t(p - 2)}{p},
\]
which corresponds to an inverse-beta-like prior supported on $(0,1)$. Under this setup, the posterior mode of $\theta$ takes the form:
\[
\hat{\theta} = (1 - \hat{\kappa}) y,
\]
where $\hat{\kappa}$ maximizes
\[
g(\kappa) = t(p - 2)\log \kappa - \kappa \|y\|^2.
\]
Solving this yields
\[
\hat{\kappa} = \min\left(1, \frac{t(p - 2)}{\|y\|^2}\right).
\]
Setting $t = 1$, we recover Stein’s positive-part shrinkage rule:
\[
\hat{\theta}_i = \left(1 - \frac{p - 2}{\|y\|^2} \right)_+ y_i.
\]
This gives us a Bayesian justification for the positive-part estimator as a posterior mode under a prior on the shrinkage weight. \citet{takada1979steins} also analyzes the corresponding posterior mean and shows that, under certain conditions, it is admissible and minimax.

\subsection{Ordered Shrinkage Priors: Motivating Examples}\label{sec:gl-shrinkage}

\citet{deaton1980empirical} considers a hierarchical Bayesian approach to polynomial regression, motivated by the empirical Bayes framework. The model is:
\[
Y_i = P(x_i) + \epsilon_i, \quad \epsilon_i \sim N(0, \sigma^2),
\]
where the regression function is a linear combination of orthonormal basis polynomials:
\[
P(x) = \sum_{j=1}^m \theta_j \psi_j(x).
\]
The polynomials $\{\psi_j(x)\}$ are assumed to be orthonormal with respect to the empirical distribution of $x_i$:
\[
\sum_{i=1}^{n} \psi_j(x_i) \psi_k(x_i) = \delta_{jk}.
\]
Let $Q$ denote the design matrix with $Q_{ij} = \psi_j(x_i)$. Under this orthonormal basis, the OLS estimate of the coefficient vector is simply $\hat{\theta} = Q^\top Y$, and the sampling distribution satisfies:
\[
\hat{\theta}_j \mid \theta_j \sim N(\theta_j, \sigma^2).
\]

Deaton's empirical Bayes method introduces a prior of the form $\theta_j \sim N(0, \sigma_j^2)$, with the constraint that $\sigma_1^2 \geq \sigma_2^2 \geq \ldots \geq \sigma_m^2$. This reflects the assumption that lower-order basis components tend to capture more signal, while higher-order ones are more likely to be negligible. Equivalently, define:
\[
z_j = \frac{\sigma^2}{\sigma_j^2 + \sigma^2}, \quad \text{so that} \quad \theta_j \mid \hat{\theta}_j \sim N\left( (1 - z_j) \hat{\theta}_j, (1 - z_j)\sigma^2 \right),
\]
with the ordering constraint:
\[
0 < z_1 \leq z_2 \leq \ldots \leq z_m \leq 1.
\]
Introducing $\kappa_j = z_j / \sigma^2$ and $\kappa_{m+1} = 1 / \sigma^2$, the marginal log-likelihood for $(\kappa_1, \ldots, \kappa_m, \kappa_{m+1})$ is (up to a constant):
\begin{equation}
\label{eq:deaton_loglik}
\log L = \left(\tfrac{1}{2} \bar{n} \right) \log \kappa_{m+1} - \tfrac{1}{2} \kappa_{m+1} W_{m+1} + \sum_{j=1}^{m} \left\{ \left(\gamma_j - \tfrac{1}{2} \right) \log \kappa_j - \tfrac{1}{2} \kappa_j \hat{\theta}_j^2 \right\},
\end{equation}
where $\bar{n}$ is a function of prior degrees of freedom, and $W_{m+1} = s + 2 / \beta$, with $s = \sum_i (y_i - Q\hat{\theta})^2$ the residual sum of squares. Following Deaton, we choose $\gamma_j = 2$ and $\beta = 5$ for all $j$.

The unconstrained maximum likelihood estimates of the $\kappa_j$ are:
\[
\hat{\kappa}_j = \frac{2\gamma_j - 1}{\hat{\theta}_j^2}, \quad \hat{\kappa}_{m+1} = \frac{\bar{n}}{W_{m+1}}.
\]
However, to enforce the monotonicity constraint $\kappa_1 \leq \kappa_2 \leq \cdots \leq \kappa_m \leq \kappa_{m+1}$, Deaton applies isotonic regression to the unconstrained $\kappa_j$ estimates (e.g., using the pool-adjacent-violators algorithm (PAVA)). Let $\hat{\kappa}_j^{\text{iso}}$ denote the constrained estimates.

The estimated shrinkage weights are then given by:
\[
\hat{z}_j = \hat{\kappa}_j^{\text{iso}} \cdot \hat{\sigma}^2, \quad \hat{\theta}_j^{\text{shrink}} = (1 - \hat{z}_j) \hat{\theta}_j.
\]

\vspace{0.5em}

\noindent Figure~\ref{fig:deaton_example} illustrates this procedure in a simulated setting that mirrors the numerical examples from \citet{deaton1980empirical}. The true signal is sparse and concentrated in the low-order terms of a Legendre polynomial basis. The observed data $Y$ are generated with noise standard deviation $\sigma = 1$, and OLS estimates are obtained. The empirical Bayes procedure then computes posterior mean shrinkage estimates after isotonic adjustment of the hypervariances. The resulting estimates exhibit sparsity aligned with the prior ordering assumption.

\begin{figure}[ht]
\centering
\includegraphics[width=0.75\textwidth]{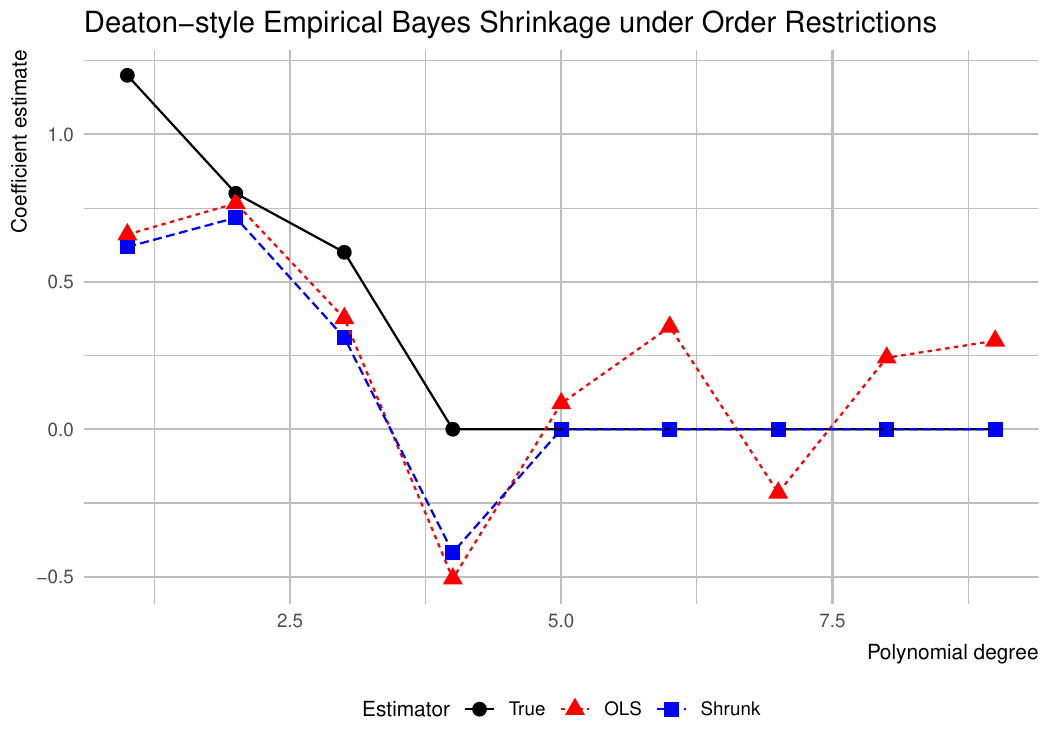}
\caption{Posterior mean estimates of $\theta_j$ using Deaton’s empirical Bayes shrinkage method with order restrictions. Lower-order terms (left) are preserved while higher-order coefficients are shrunk aggressively toward zero.}
\label{fig:deaton_example}
\end{figure}

\paragraph{Model class choice: sparse monotone versus Sobolev.}
The ordered-variance prior in Deaton's construction implicitly targets a function class in which the orthogonal-series coefficients are concentrated in low-order terms and then decay along the natural ordering of the basis. This sparse ordered viewpoint directly matches the structure in \citet{deaton1980empirical} and \citet{xu2007some}, and it aligns with polynomial (or more generally orthogonal-series) regression where it is natural to expect higher-order components to contribute less signal. It also leads to a clean risk story: the empirical Bayes estimator of the variance profile via isotonic regression is tailored to this ordered structure, and lower bounds depend on the exact formulation of the ordered parameter space.
Sobolev balls provide a complementary, smoothness-driven alternative, but they require specifying a smoothness index $\beta$, and in a Bayesian treatment they naturally raise the question of why one would not instead adopt Gaussian process priors (or related spline/RKHS priors) as the primary modeling device. That direction is valuable but brings different technical machinery and shifts the focus away from the ordered shrinkage profiles inspired by \citet{deaton1980empirical} and \citet{xu2007some} and their applied interpretability in polynomial regression settings. For these reasons, we take the sparse monotone class as the primary target of our discussion and theory, with Sobolev-type smoothness as a secondary perspective.



\subsubsection{Shrinkage in Infinite Gaussian Means Models}

We now consider the problem of estimating a sequence of normal means in a hierarchical setting. This formulation captures the structure of both the James--Stein problem and empirical Bayes estimation under global-local shrinkage.

Suppose the observations follow the hierarchical model:
\begin{align*}
  y_i &= \theta_i + \epsilon_i, \quad \epsilon_i \sim N(0, \sigma^2), \\
  \theta_i &= \mu + \tau u_i, \quad u_i \sim N(0, 1),
\end{align*}
with possibly heterogeneous variances if we allow $\tau = \tau_i$ to vary across components. This corresponds to a classical variance components model, sometimes referred to as the one-way random effects or ANOVA model.

The classical estimation of $(\sigma^2, \tau^2)$ is known to be challenging. In particular, the Neyman--Scott paradox \citep{neyman1948consistent} shows that the MLE of $\sigma^2$ is inconsistent when the number of parameters grows with the sample size. In small-sample settings, standard estimators may yield negative variance estimates, as shown by \citet{hill1965inference}. These pathologies motivate Bayesian shrinkage methods.

A fully Bayesian approach imposes priors on the variance components. When $\tau_i$ is allowed to vary across $i$, this leads to a global-local model. In infinite-dimensional or growing-$p$ settings, the prior decay schedule for the local variances can strongly affect statistical performance, which motivates procedures that estimate (or adapt to) the local variance profile from data rather than fixing a single parametric decay form.

\begin{remark}
The optimal local variance $\tau_i^2$ can be estimated using the least concave majorant (LCM) of the CUSUM of $X_{in}^2 - \sigma^2/n$, which serves as a form of isotonic regression. This perspective is closely related to Deaton's order-restricted empirical Bayes construction for polynomial regression \citep{deaton1980empirical} and to subsequent developments based on LCM/PAVA-type estimators \citep{xu2007some}.
\end{remark}

In his PhD thesis \cite{xu2007some} showed how least concave majorants (LCM) can be used to estimate a sequence of prior variances $\tau_i^2$ from observed data $X_i^2$. The method constructs the LCM of the cumulative sum of squared signals, effectively performing an adaptive isotonic regression on the variance sequence. This procedure recovers structure in sparse or smoothly decaying sequences, and is particularly useful for modeling high-dimensional global-local priors.

Figure~\ref{fig:xu_variance_combined} reproduces two illustrations from \citet{xu2007some}. The left panel (Example 2.3) shows the log-CUSUM of $X_i^2$ (dots) for $r=1,\ldots,1005$, along with its least concave majorant (blue line), providing a smoothed approximation of signal accumulation. The right panel (Example 2.4) compares three quantities on a log-scale: (i) true variance decay $\tau_i^2 = 1/i^2$ (dashed), (ii) noisy variance proxies $X_i^2$ (dots), and (iii) refined estimates from the slope of the LCM (blue line). The LCM estimates exhibit superior stability in estimating $\tau_i^2$.

\begin{figure}[ht!]
\centering
\begin{subfigure}[b]{0.48\linewidth}
  \includegraphics[width=\linewidth]{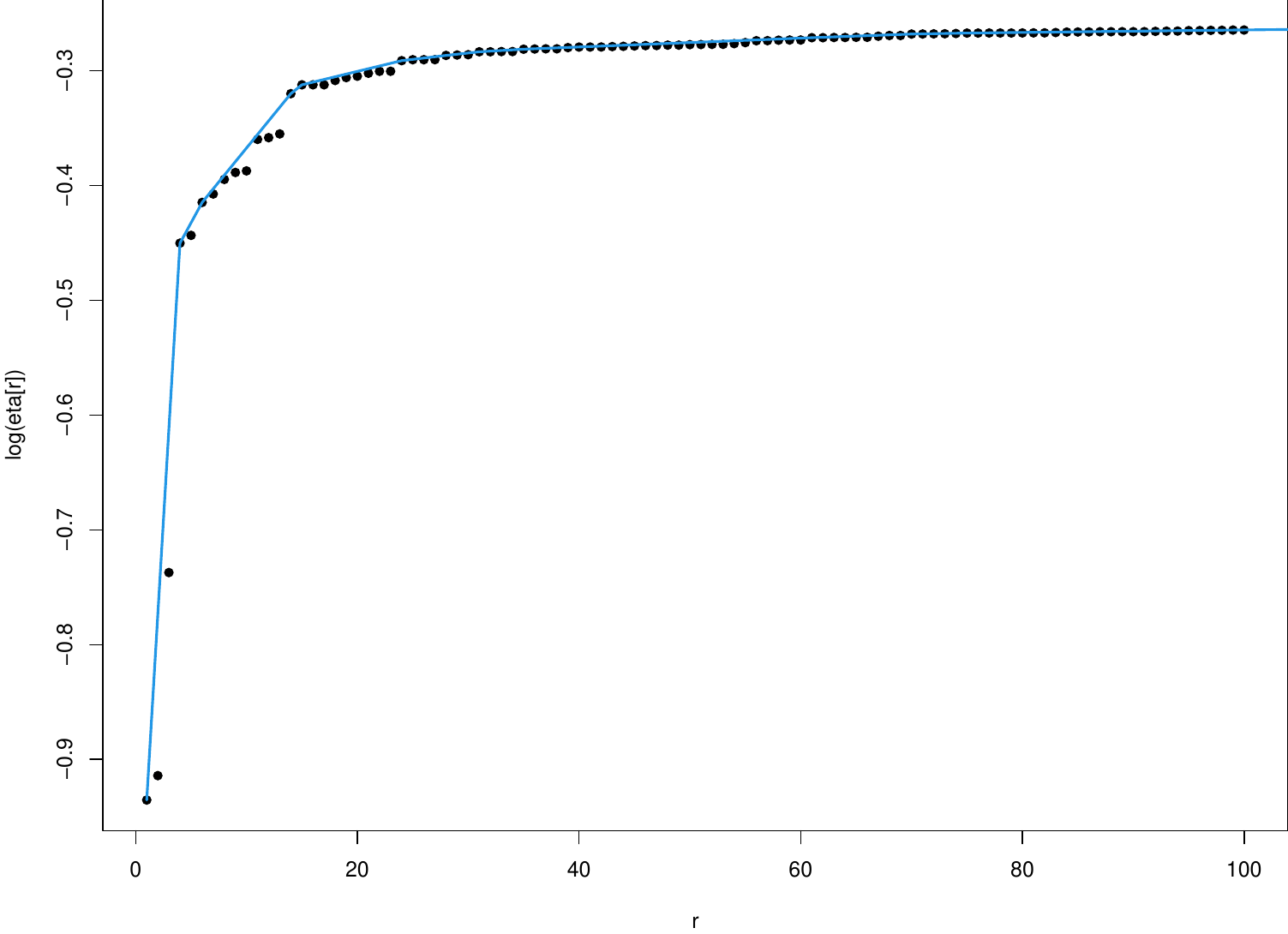}
  \caption{CUSUM of $X_i^2$ (dots) and LCM (blue line)}
\end{subfigure}
\hfill
\begin{subfigure}[b]{0.48\linewidth}
  \includegraphics[width=\linewidth]{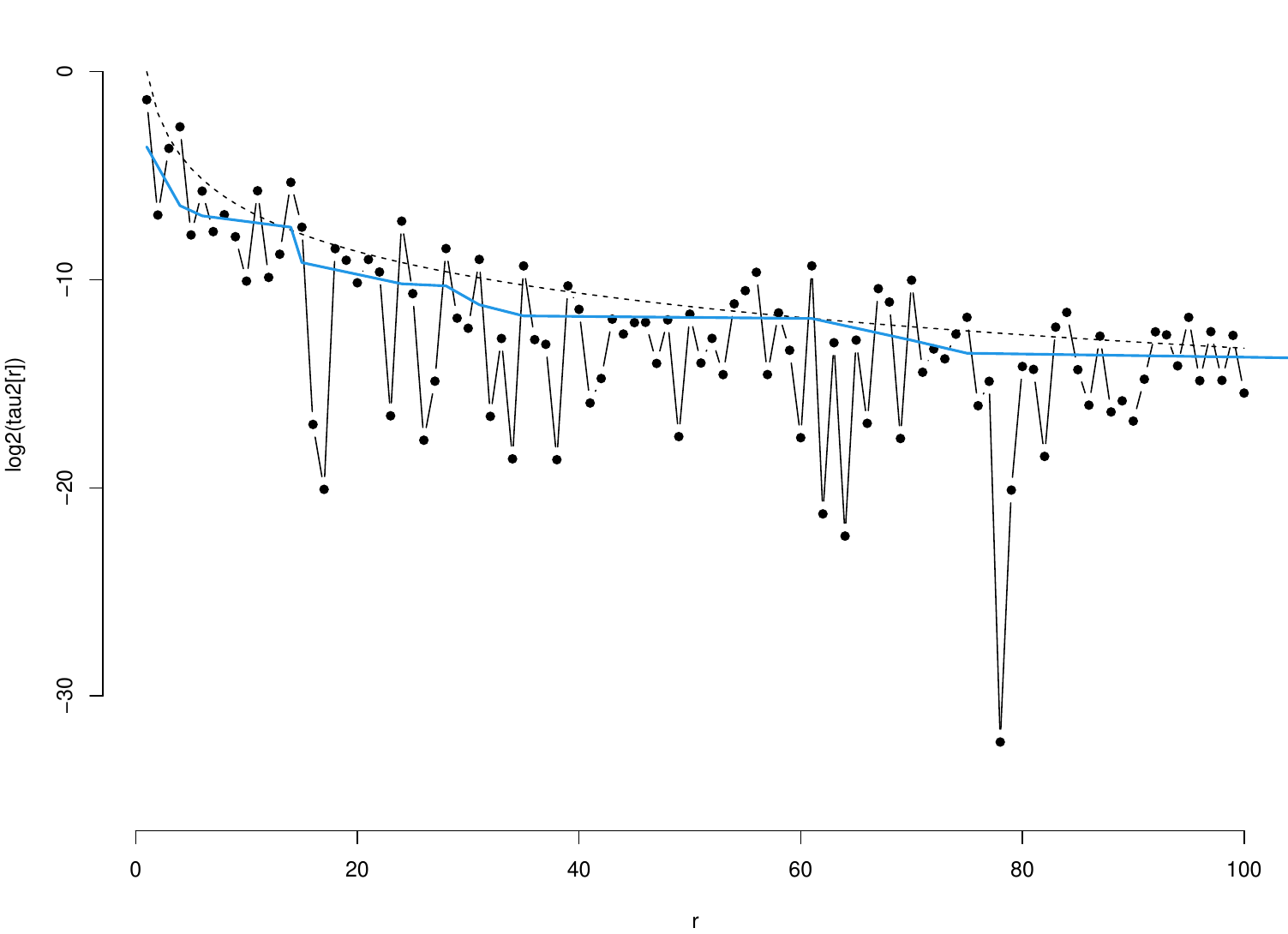}
  \caption{True vs. empirical vs. LCM-refined $\tau_i^2$ (log scale)}
\end{subfigure}
\caption{Illustration of Xu's variance component estimation method via least concave majorants (LCM). Adapted from Figures 2.5 and 2.6 of \citet{xu2007some}.}
\label{fig:xu_variance_combined}
\end{figure}

\paragraph{Infinite Gaussian Means}

The Bayes estimate of the sequence is given by:
\[
\hat{\theta}_{in} = \frac{\hat{\tau}_i^2}{\hat{\tau}_i^2 + \frac{\sigma^2}{n}}X_{in} = \hat{w}_i X_{in}
\]
The optimal weights are  $1-z_i$, then we order and estimate the $\hat{z}_i$ values.
\[
E\left(X_{in}^2 - \frac{\sigma^2}{n}\right) = \tau_i^2
\]
Least concave majorant (LCM) of the cumulative sum (CUSUM) of $X_{in}^2 - \frac{\sigma^2}{n}$ provides the estimator $\hat{\tau}_i^2$. Let $\eta_r = \sum_{i=1}^r X_i^2$ and consider the piecewise-linear interpolation of its LCM. Define the slope at each $i$ as
\begin{equation*}
\hat{\tau}_i^2 = \left(\min_{1 \leq h \leq i} \max_{i \leq j \leq p} \frac{1}{j - h + 1} \sum_{l=h}^j \left(X_l^2 - \frac{\sigma^2}{n}\right) \right)_+.
\end{equation*}
This isotonic estimator shrinks noisy $X_i^2$ estimates toward the underlying monotone structure implied by the decay of the true signal.

\section{Regression and Regularization}
\label{sec:global_local_framework}

We now connect the ordered-variance empirical Bayes perspective to classical regularization in regression. The key point is that many regularization rules become simple coordinatewise shrinkage rules after transforming to a canonical (orthogonal) basis, and these shrinkage factors can be interpreted as (possibly structured) variance components in a hierarchical Gaussian model.

Consider the sparse normal means problem, or the Gaussian compound decision problem: $y_i = \theta_i + \epsilon_i$ where
$$
\theta_i \mid \kappa_i  \sim \NormRV\left ( 0 , \frac{1- \kappa_i}{ \kappa_i} \right). 
$$
This parameterization yields a broad family of shrinkage priors through priors on $\kappa_i$ (including global--local choices such as the horseshoe); see, e.g., \citet{bhadra2016default} for discussion of default behavior for nonlinear functionals.

\subsection{Generalized Ridge Regression in the Canonical Basis}

We revisit the ridge regression framework from the perspective of the canonical (rotated) coordinate system, consistent with the singular value decomposition $X = U D W^\top$ and rotated coefficients $\alpha = W^\top \beta$, and $Z = U^\top Y$, $\alpha = V^\top \beta$, and $\nu = U^\top \epsilon$. In this basis, the regression model becomes:
\[
Y = X \beta + \epsilon = U D \alpha + \epsilon \Rightarrow U^\top Y = Z = D \alpha + \nu, 
\]
and projecting onto $U$, we obtain:
\[
\hat{\alpha}_i = \alpha_i + \epsilon_i, \quad \epsilon_i \sim N(0, \sigma^2/d_i^2), \quad 1 \le i \le p, 
\]
where $d_i$ is the $i$th singular value and $\hat{\alpha}_i$ is the least-squares estimate of $\alpha_i$. We now consider the following hierarchical model:
\begin{align*}
z_i \mid \alpha_i, \sigma^2 &\sim \mathcal{N}(d_i \alpha_i, \sigma^2), \quad 1 \leq i \leq p \\
z_i \mid \sigma^2 &\sim \mathcal{N}(0, \sigma^2), \quad p+1 \leq i \leq n. 
\end{align*}

Now suppose, we place independent priors $\alpha_i \sim N(0, v_i)$ on the coefficients in this rotated basis. The posterior mean under this Gaussian model yields a shrinkage estimator of the form:
\begin{equation}
\alpha_i^* = \kappa_i \hat{\alpha}_i, \quad \text{with} \quad \kappa_i = \frac{d_i^2 v_i}{\sigma^2 + d_i^2 v_i} = \frac{d_i^2}{d_i^2 + k_i},
\end{equation}
where $k_i = \sigma^2 / v_i$ is the generalized ridge penalty for the $i$th component. When all $k_i = k$, this recovers standard ridge regression.

\[
p(\alpha_i\mid \kappa_i) = \sqrt{\dfrac{\lambda_i(1-\kappa_i)}{2\pi \kappa_i}}\exp\left(-\dfrac{\lambda_i(1-\kappa_i)}{2\kappa_i}\alpha_i^2\right) \sim \NormRV\left(0,\dfrac{\kappa_i}{\lambda_i(1-\kappa_i)}\right).
\]
The log-likelihood is
\[
\log p(D\mid \alpha, \sigma^2) = c - \dfrac{n}{2}\log \sigma^2 - \dfrac{1}{2\sigma^2}||Y-X\alpha||_2^2 =  c - \dfrac{n}{2}\log \sigma^2 + \dfrac{1}{2} \sum_{i=1}^{p}\left\{\log(1-\kappa_i) + \dfrac{\kappa_i\lambda_i\hat \alpha_i^2}{\sigma^2}\right\}.
\]
For regularization priors, $w_i \sim \text{Beta}(a,b)$. Note that when $a=b=1/2$ this leads to the Horseshoe model. Let $t_i = \sqrt{\lambda_i}\hat \alpha_i/\sigma$. Many authors use deterministic $W(t)$, e.g., Empirical Bayes or Type II maximum likelihood \citep[see e.g.][]{polson2012good}.

\begin{figure}[ht]
\centering
\includegraphics[width=0.7\textwidth]{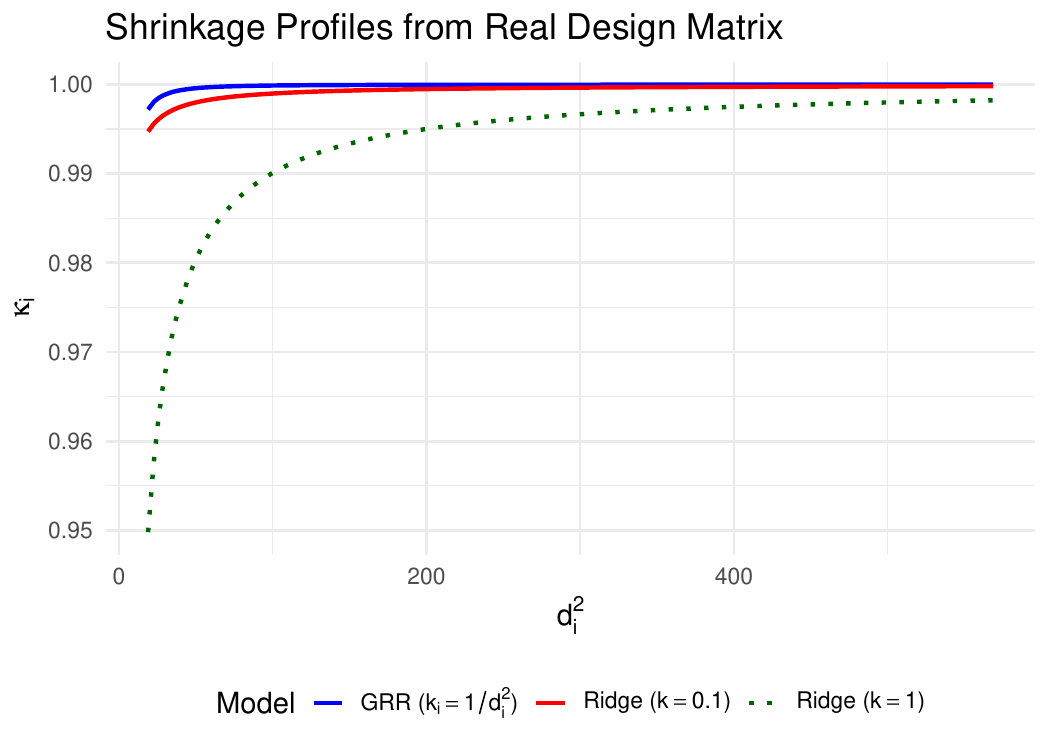}
\caption{Shrinkage profiles $\kappa_i = d_i^2 / (d_i^2 + k_i)$ based on singular values from a real design matrix. Ridge uses a fixed penalty $k=1$ and $k_i=0.1$ (blue, dashed), and the adaptive choice uses $k_i = 1/d_i^2$ (red, solid).}
\label{fig:shrinkage_profiles}
\end{figure}

Figure~\ref{fig:shrinkage_profiles} compares shrinkage weights $\kappa_i = d_i^2 / (d_i^2 + k_i)$ under three regularization schemes: classical ridge regression with fixed penalties $k = 1$ and $k = 0.1$, and a generalized ridge rule with data-adaptive weights $k_i = 1/d_i^2$. The latter mimics global-local shrinkage by applying minimal shrinkage to large singular components and aggressive shrinkage to smaller ones.

\begin{remark}
Goldstein and Smith \citep{goldstein1974ridgetype} showed that if the smallest singular value satisfies \( \min d_i^2 > (c + 4)/2 \sum d_j^{-2} \), then the above estimator dominates least squares for all \( i \) in componentwise risk. This is a stronger guarantee than the usual Stein result, which only applies to the total risk.
\end{remark}

In other words, while the James--Stein estimator also shrinks the OLS coefficients toward zero, it does so uniformly across all coordinates through a global multiplier that depends on the total squared norm $\|\hat{\alpha}\|^2$. In contrast, the GRR estimator performs shrinkage independently along each coordinate, based on local information---specifically, the corresponding singular value $d_i$ and penalty $k_i$. This local, componentwise shrinkage anticipates the structure of modern global-local shrinkage priors, where each parameter receives its own adaptive regularization governed by local scale terms. In both settings, the goal is to suppress estimation error more aggressively in directions of low signal-to-noise, but GRR achieves this in a fully decoupled fashion, which enables componentwise risk improvements that the James--Stein estimator cannot guarantee.

\subsubsection{Adaptive Selection of Shrinkage Weights}

In global-local shrinkage models and generalized ridge regression (GRR), the strength of regularization is governed by local penalties $k_i$, which correspond to prior precision parameters in the canonical regression basis. Selecting appropriate values for these $k_i$'s is essential to achieve optimal bias-variance tradeoffs.

A common approach is to use marginal maximum likelihood or Type II maximum likelihood estimation. Suppose we observe rotated data $z_i = d_i \hat{\alpha}_i$, with marginal distribution:
\begin{equation*}
p(z_i \mid k_i, \sigma^2) \sim N\left(0, \sigma^2\left(1 + \frac{d_i^2}{k_i}\right)\right).
\end{equation*}
Let $\phi_i^2 = \sigma^2(1 + d_i^2 / k_i)$. By reparametrizing and maximizing the likelihood with respect to $k_i$, we obtain the empirical Bayes estimate:
\begin{equation*}
\hat{k}_i = \frac{d_i^2 \sigma^2}{z_i^2 - \sigma^2}, \quad \text{provided that } z_i^2 > \sigma^2.
\end{equation*}
This yields a data-adaptive shrinkage rule, where coordinates with weak signal ($z_i^2 \leq \sigma^2$) receive effectively infinite penalty ($k_i = \infty$), corresponding to complete shrinkage. Stronger signals receive smaller penalties and are retained.

\begin{remark}
This thresholding rule mimics the behavior of the non-negative garotte estimator proposed by \citet{breiman1995better}, wherein coordinates are retained or eliminated based on their marginal evidence. It also aligns with modern empirical Bayes variable selection ideas.
\end{remark}

\subsubsection{Comparison with Classical Regularization Methods}
\citet{frank1993statistical} showed that many regularization methods can be expressed as:
\[
\hat{\beta}^M = \sum_{j=1}^p \kappa_{j}^M \hat{\alpha}_j w_j
\]

Ridge regression has $\kappa_j^{RR} = d_j^2 / (Nu + d_j^2)$ for a fixed regularization parameter $Nu$. $K$-component PCR has:
\[
\kappa_{jK}^{PCR} = 
\begin{cases}
1, & d_j^2 \geq d_K^2 \\
0, & d_j^2 < d_K^2
\end{cases}
\]

The posterior mean under the $g$-prior corresponds to $\kappa_j^g = g/(1+g)$, shrinking the solution vector along all eigen-directions by a common factor. To enhance flexibility over the standard $g$-prior, mixtures of $g$-priors have been proposed \citep{liang2008mixtures}, which improve variable selection performance through more adaptive shrinkage mechanisms. For high-dimensional settings where $p>n$, \citet{maruyama2011fully} developed a modified $g$-prior that enables Bayesian variable selection even when conventional approaches fail.

These all fit into a general Bayesian prior $(\alpha|\sigma^2,\tau^2,\Lambda) \sim N(0, \sigma^2\tau^2\Lambda)$ where $\Lambda = \text{diag}(\lambda_1^2,\ldots,\lambda_p^2)$. The posterior mean has shrinkage factors:
\[
\kappa_j = \frac{\tau^2\lambda_j^2d_j^2}{1 + \tau^2\lambda_j^2d_j^2}
\]

\subsection{Generalization to large \texorpdfstring{$p$}{p} cases}
Suppose now that the design matrix $X$ is of rank $r<p$ and has singular-value decomposition $X = U D W'$ with $D = \mbox{diag}(d_1, \ldots, d_r)$, again ordered from largest ($d_1$) to smallest ($d_r$).   The approach of the previous section works just as before, with no essential modification:
\begin{eqnarray*}
(\hat{\alpha} \mid \alpha, \sigma^2) &\sim& N(\alpha, \sigma^2 D^{-2}) \\
(\alpha \mid \sigma^2, \tau^2, \Lambda) &\sim& N(0, \sigma^2 \tau^2 \Lambda) \\
\lambda_j^2 &\sim& p(\lambda_j^2) \\
(\sigma^2, \tau^2) &\sim& p(\sigma^2, \tau^2) \, ,
\end{eqnarray*}
where $\alpha = W' \beta$ and $\hat{\alpha}$ is the corresponding OLS estimate.  Instead of a $p$-dimensional vector to estimate, we now have an $r$-dimensional one.  Moreover, because we have orthogonalized the coefficients, the elements of $\alpha$ are conditionally independent in the posterior distribution, given $\tau^2$ and $\sigma^2$.  We are faced with a simple normal-means problem, with the only complication being that the singular values $d_j$ enter the likelihood.

\subsection{Bayesian interpretation}
These four procedures differ only in the way that they scale the OLS estimates for the regression parameter in the orthogonal coordinate system defined by $W$.  It is therefore natural to consider them as special cases of an encompassing local-shrinkage model along the lines of the previous sections. Table~\ref{tab:shrinkage_weights} summarizes the shrinkage weight functions $W(t)$ for several classical regularization methods, showing how each method maps the signal-to-noise ratio $t$ to a shrinkage factor.

\begin{table}[ht!]
\centering
\footnotesize{
\begin{tabular}{|l|l|}
\hline
\textbf{Method} & \textbf{Shrinkage Weight $W(t)$} \\
\hline
GRN (Hoerl and Kennard, 1970) & $W(t) = (1 + t^{-2})^{-1}$ \\
GRI1 (Hoerl and Kennard, 1970) & $W(t) = \left\{1 + \left(1 + t^{-2}\right)^2 / t^2\right\}^{-1}$ \\
GRI (Hemmerle, 1975) & $W(t) = \left[\dfrac{1 - \sqrt{1 - 4t^{-2}}}{2} t^2\right]^{-1}$ for $t^2 > 4$, else $W(t) = 0$ \\
GRP (Mallows, 1973) & $W(t) = 1$ for $t^2 > 2$, else $W(t) = 0$ \\
GRC (Mallows, 1973) & $W(t) = 1 - t^{-2}$ for $t^2 > 1$, else $W(t) = 0$ \\
\hline
\end{tabular}}
\caption{Shrinkage weights $W(t)$ for various regularization methods.}
\label{tab:shrinkage_weights}
\end{table}

Having established the connection between ordered-variance priors and classical regularization in canonical coordinates, we now formalize the risk-theoretic properties of the isotonic empirical Bayes estimator. The next section provides frequentist risk bounds showing that this estimator achieves near-minimax optimality over sparse ordered model classes.

\section{Risk bounds for isotonic empirical Bayes global--local shrinkage}
\label{sec:risk_bounds}

\paragraph{Estimator in brief (practice versus proof device).}
The practical isotonic empirical Bayes procedure estimates an ordered local-variance (or shrinkage-weight) profile by maximizing the marginal likelihood under a monotonicity constraint, and then plugs the constrained estimate into the Gaussian posterior mean shrinkage rule. To keep the risk analysis fully rigorous (avoiding dependence between an estimated variance profile and the same Gaussian coordinates being shrunk), the theorems below analyze a randomized, cross-fit analogue based on a Gaussian cloning/splitting construction. This cross-fit estimator is used as a proof device; it implements the same monotone-variance estimation idea but ensures the fitted variance profile is independent of the Gaussian sample to which shrinkage is applied.

We consider the Gaussian sequence model
\begin{equation}
Y_i = \theta_i + \varepsilon_i, \qquad 
\varepsilon_i \stackrel{\text{i.i.d.}}{\sim} N(0, \sigma^2/n), 
\quad i = 1,\dots,p,
\label{eq:sequence}
\end{equation}
where $\sigma^2 > 0$ is treated as known in this section and $n$ denotes the effective sample size.
An extension to unknown $\sigma^2$ via cross-fitting is given in Theorem~\ref{thm:unknown_sigma}.

Write $\lambda=\sigma^2/n$ for the per-observation noise variance. Throughout this section and the appendix proofs, we also use $\nu$ for the variance of Gaussian inputs to the isotonic regression step; under the cloning construction introduced below, $\nu = 2\lambda = 2\sigma^2/n$.

To make the empirical Bayes shrinkage analysis fully rigorous without delicate dependence issues,
we use a standard Gaussian ``cloning'' trick. Let $Z_i\stackrel{\text{i.i.d.}}{\sim}N(0,\lambda)$ be independent of $(\varepsilon_i)$ and define
\[
Y_i^{(+)} := Y_i + Z_i, \qquad Y_i^{(-)} := Y_i - Z_i,
\]
so that $Y^{(+)}$ and $Y^{(-)}$ are independent and each satisfies $Y_i^{(\pm)}\sim N(\theta_i,2\lambda)$.

Define the monotone variance proxy sequence $V_i := \theta_i^2$.
To control the weighted loss \(\sum_i(\widehat V_i-V_i)^2/(V_i+2\lambda)\) under the heteroskedastic
squared-Gaussian noise, we estimate \(V\) from \(Y^{(+)}\) using a cross-fit (sample-split) isotonic
procedure.

\begin{assumption}[Binning-scale regularity]
\label{assump:B1}
Let \(\nu>0\) and let \(V=(V_1,\dots,V_p)\) be a nonincreasing, nonnegative sequence with \(V_1\le R\).
Let \((t_m)_{m=0}^M\) be dyadic thresholds with \(t_m=2^m\nu\), and let \((I_m)_{m=0}^{M-1}\) be the contiguous bins used
to define the piecewise-constant weights \(w_i:=t_{m(i)}^{-1}\) (where \(m(i)\) is the unique index with \(i\in I_{m(i)}\)).
We assume the bins satisfy, for all \(m\in\{0,\dots,M-1\}\) and all \(i\in I_m\),
\begin{align*}
t_m \ \le\ V_i+\nu \ <\ t_{m+1}=2t_m.
\tag{B1}
\end{align*}
\end{assumption}

\begin{assumption}[Margin from dyadic thresholds]\label{assump:margin_dyadic}
Fix $\kappa\in(0,1/4)$. For each $i\le p$, let $m^\star(i)$ be the (deterministic) dyadic index such that
$t_{m^\star(i)}\le V_i+\nu<t_{m^\star(i)+1}=2t_{m^\star(i)}$.
Assume that for all $i\le p$,
\begin{equation}
\label{eq:margin_dyadic}
t_{m^\star(i)}(1+\kappa)\ \le\ V_i+\nu\ \le\ t_{m^\star(i)+1}(1-\kappa).
\end{equation}
\end{assumption}

\begin{remark}[Practical interpretation of Assumptions~\ref{assump:B1} and~\ref{assump:margin_dyadic}]
\label{rem:practical_B1}
Assumption~\ref{assump:B1} requires that the pilot isotonic fit assigns each coordinate to the ``correct'' dyadic scale. In practice, this condition holds with high probability when the signal-to-noise ratio is not too extreme and the true variance profile is well-separated from dyadic boundaries. The margin condition~\ref{assump:margin_dyadic} formalizes this separation: if $V_i + \nu$ stays away from the dyadic thresholds by a factor of at least $\kappa$, the pilot fit will correctly identify the scale.

Concretely, practitioners can verify these conditions by checking: (i)~the sample size $n$ is large enough that the noise variance $\nu = 2\sigma^2/n$ is small relative to the signal magnitudes $V_i = \theta_i^2$, and (ii)~the signal profile does not cluster at dyadic boundaries. For most smooth or sparse signals arising in polynomial regression or spectral methods, these conditions are satisfied. Proposition~\ref{prop:margin_implies_B1} shows that the binning event \textup{(B1)} holds with probability at least $1 - O(p^{-3})$ when the margin parameter satisfies $\kappa^2 \nu/(R+\nu) \gtrsim \log p$.
\end{remark}

\paragraph{Notation correspondence.}
In Lemma~\ref{lem:isotonic_weighted_main}, we use \(\nu\) for the noise variance of the input Gaussians.
Under the cloning construction above, the clones \(Y^{(+)}\) and \(Y^{(-)}\) have variance \(2\lambda=2\sigma^2/n\),
so we apply Lemma~\ref{lem:isotonic_weighted_main} with \(\nu:=2\lambda=2\sigma^2/n\).

We then apply the resulting shrinkage weights to the independent clone $Y^{(-)}$ and define
\begin{equation}
\widehat{\theta}_i
:=
\frac{\widehat{V}_i}{\widehat{V}_i + 2\sigma^2/n}\, Y_i^{(-)}.
\label{eq:postmean}
\end{equation}
This is a randomized, cross-fit estimator used as a proof device: it depends on auxiliary Gaussian randomness
through the cloning/splitting construction (which makes the variance-profile estimate \(\widehat V\) independent of
the Gaussian sample \(Y^{(-)}\) to which shrinkage is applied). It performs coordinatewise shrinkage; it does not
impose a hard cutoff at an index $s$.

We now define the parameter class over which our risk bounds hold.

\begin{definition}[Sparse ordered class]
\label{def:class}
For integers $1 \le s \le p$ and constant $R > 0$, define the sparse monotone class
\begin{equation}
\Theta_{\downarrow}(s,R)
=
\left\{
\theta \in \mathbb{R}^p :
\exists k \le s \text{ such that }
\theta_1^2 \ge \theta_2^2 \ge \cdots \ge \theta_k^2 \ge 0,\quad
\theta_1^2 \le R,
\;\;
\theta_i = 0 \text{ for } i > k
\right\}.
\label{eq:class}
\end{equation}
\end{definition}

\begin{theorem}[Adaptive near-minimaxity (up to logarithmic factors) of isotonic empirical Bayes shrinkage]
\label{thm:main}
Let $\widehat{\theta}$ denote the randomized, cross-fit isotonic empirical Bayes shrinkage estimator
defined in \eqref{eq:postmean} (based on the Gaussian cloning/splitting construction).
The risk bound below is stated conditionally on the pilot/binning event \textup{(B1)} from
Assumption~\ref{assump:B1}.
Then there exists a universal constant $C>0$ such that
\begin{equation}
\sup_{\theta \in \Theta_{\downarrow}(s,R)}
\mathbb{E}_{\theta}\!\left[
\left\|
\widehat{\theta} - \theta
\right\|_2^2
\,\middle|\,
\textup{(B1)}
\right]
\;\le\;
C \, s \frac{\sigma^2}{n}\log\!\Big(\frac{ep}{s}\Big)\,
\log\!\Big(1+\frac{Rn}{\sigma^2}\Big).
\label{eq:risk}
\end{equation}
\noindent The bound matches the minimax benchmark over $\Theta_{\downarrow}(s,R)$ up to the displayed logarithmic factors;
all dependence on $R$ is through the explicit factor $\log(1+Rn/\sigma^2)$.
\noindent Moreover, if $\sigma^2/n \le R$ then there exists a universal constant $c>0$ such that
\[
\inf_{\tilde{\theta}}
\sup_{\theta\in\Theta_{\downarrow}(s,R)}
\mathbb{E}_\theta\|\tilde{\theta}-\theta\|_2^2
\ \ge\
c\, s\frac{\sigma^2}{n},
\]
so the minimax risk over $\Theta_{\downarrow}(s,R)$ is at least of order $s(\sigma^2/n)$ (Lemma~\ref{lem:lower_main}).

\vspace{0.25em}

\noindent If in addition Assumption~\ref{assump:margin_dyadic} holds and the dyadic-margin scaling satisfies
\(\kappa^2\nu/(R+\nu)\gtrsim \log p\) (with \(\nu:=2\sigma^2/n\) as in the cloning construction),
then the conditional risk bound \eqref{eq:risk} admits an end-to-end (unconditional) version:
there exists a universal constant $C'>0$ such that
\[
\sup_{\theta \in \Theta_{\downarrow}(s,R)}
\mathbb{E}_{\theta}\!\left[
\left\|
\widehat{\theta} - \theta
\right\|_2^2
\right]
\ \le\
C' \, s \frac{\sigma^2}{n}\log\!\Big(\frac{ep}{s}\Big)\,
\log\!\Big(1+\frac{Rn}{\sigma^2}\Big)
\;+\;
C'\Big(sR+\frac{p\sigma^2}{n}\Big)p^{-3/2}.
\]
\end{theorem}

\begin{remark}[Conditional nature of the bound]
\label{rem:b1_end_to_end}
Theorem~\ref{thm:main} states a risk bound conditional on the pilot/binning event \textup{(B1)}.
Assumption~\ref{assump:B1} is a pilot-localization property ensuring that the data-dependent bins used for weighting
are on the correct dyadic scale relative to the unknown $V_i+\nu$.
To obtain an end-to-end unconditional guarantee, one must also prove that \textup{(B1)} holds with high probability
uniformly over $\theta\in\Theta_{\downarrow}(s,R)$ (or, alternatively, add an explicit remainder term controlling the
contribution of \(\textup{(B1)}^c\)).
\end{remark}

\begin{corollary}[Low signal-to-noise simplification]
\label{cor:main_low_snr}
If \(Rn/\sigma^2\le 1\), then there exists a universal constant $C_{\mathrm{snr}}>0$ such that
\[
\sup_{\theta \in \Theta_{\downarrow}(s,R)}
\mathbb{E}_{\theta}\!\left[\left\|\widehat{\theta} - \theta\right\|_2^2 \,\middle|\, \textup{(B1)}\right]
\ \le\
C_{\mathrm{snr}} \, s \frac{\sigma^2}{n}\log\!\Big(\frac{ep}{s}\Big).
\]
\end{corollary}
\begin{proof}
Since \(Rn/\sigma^2\le 1\), we have \(\log(1+Rn/\sigma^2)\le \log 2\), which is absorbed into the universal constant.
\end{proof}

The frequentist risk bounds established above imply corresponding posterior concentration guarantees. In the next section, we show that the Gaussian posterior induced by isotonic empirical Bayes shrinkage contracts around the true parameter at the same near-minimax rate.

\section{Concentration Results}
\label{sec:concentration_results}

\subsection{Posterior contraction in \texorpdfstring{$\ell_2$}{l2} norm}

\begin{theorem}[Posterior contraction for isotonic empirical Bayes shrinkage]
\label{thm:main_contr}
Consider the Gaussian sequence model
\[
Y_i=\theta_{0,i}+\varepsilon_i,
\varepsilon_i \stackrel{\mathrm{i.i.d.}}{\sim} N(0,\sigma^2/n),
\quad i=1,\dots,p,
\]
with known $\sigma^2>0$. Let $\Pi(\cdot\mid Y)$ denote the (data-dependent) Gaussian posterior induced by the
isotonic empirical Bayes prior
$\theta_i\mid\tau_i^2\sim N(0,\tau_i^2)$ with $\tau_1^2\ge\cdots\ge\tau_p^2\ge 0$,
where $\widehat{\tau}^2$ is estimated by marginal maximum likelihood under the monotonicity constraint.
Write $\widehat{m}$ and $\widehat{V}$ for the posterior mean and covariance.

Assume the posterior second moment satisfies, uniformly over a parameter class $\Theta$,
\begin{equation}
\label{eq:moment_maintext}
\sup_{\theta_0\in\Theta}\;
\mathbb{E}_{\theta_0}\!\left[\,
\|\widehat{m}-\theta_0\|_2^2+\mathrm{tr}(\widehat{V})
\,\right]
\;\le\; C\,\varepsilon_n^2,
\end{equation}
for some rate $\varepsilon_n\downarrow 0$ and constant $C>0$.
Then for any sequence $M\to\infty$,
\[
\sup_{\theta_0\in\Theta}\;
\mathbb{E}_{\theta_0}\Big[
\Pi\big(\|\theta-\theta_0\|_2 > M\varepsilon_n \,\big|\, Y\big)
\Big]
\;\longrightarrow\;0.
\]

In particular, if \eqref{eq:moment_maintext} holds over the ordered sparse class $\Theta_{\downarrow}(s,R)$ with
\[
\varepsilon_n^2 \asymp s\,\frac{\sigma^2}{n},
\]
then the isotonic EB posterior contracts at this rate. Likewise, if \eqref{eq:moment_maintext} holds over the
Sobolev ellipsoid $\mathcal{W}(\beta,R)$ with $\beta>1/2$ with
\[
\varepsilon_n^2 \asymp
R^{\frac{2}{2\beta+1}}
\Big(\frac{\sigma^2}{n}\Big)^{\frac{2\beta}{2\beta+1}}.
\]
then the posterior contracts at this rate.
\end{theorem}

\begin{proof}
See Appendix~\ref{sec:contraction} for the full proof.
\end{proof}

\begin{remark}[Verification of the second-moment condition]
The assumption~\eqref{eq:moment_maintext} is verified by Theorem~\ref{thm:main} for the ordered sparse class $\Theta_{\downarrow}(s,R)$ (under the pilot/binning regularity condition \textup{(B1)}), and by Corollary~\ref{cor:sobolev} for the Sobolev ellipsoid $\mathcal{W}(\beta,R)$. The trace bound $\mathbb{E}_{\theta_0}[\mathrm{tr}(\widehat{V})]$ is of the same order as the mean-squared error bound, since each posterior variance $\widehat{v}_i\le \sigma^2/n$ and the effective dimension is controlled by the shrinkage factors.
\end{remark}

\subsection{Contraction in prediction norm}

\begin{corollary}[Posterior contraction in prediction norm]
\label{cor:main_pred}
Consider the linear regression model
\[
Y = X\beta_0 + \varepsilon,
\qquad
\varepsilon \sim N(0,\sigma^2 I_n),
\]
with fixed design matrix $X\in\mathbb{R}^{n\times p}$ of rank $r$.
Let $\Pi(\cdot\mid Y)$ denote the isotonic empirical Bayes posterior on $\beta$
induced by applying isotonic EB shrinkage to the canonical coefficients obtained
from the singular value decomposition of $X$.

Suppose the corresponding canonical mean vector
$\theta_0 = DV^\top\beta_0 \in \mathbb{R}^r$
belongs to an ordered sparse class $\Theta_{\downarrow}(s,R)$
(or, more generally, to a Sobolev ellipsoid $\mathcal W(\beta,R)$).
Then the posterior contracts in prediction norm at the corresponding canonical rate:
\[
\sup_{\beta_0:\, DV^\top\beta_0\in\Theta_{\downarrow}(s,R)}
\mathbb{E}_{\beta_0}
\Big[
\Pi\big(\|X(\beta-\beta_0)\|_2 > M\varepsilon_n \mid Y\big)
\Big]
\;\longrightarrow\; 0,
\qquad M\to\infty,
\]
with
\[
\varepsilon_n^2
\asymp
s\,\frac{\sigma^2}{n},
\]
and, for $\theta_0\in\mathcal W(\beta,R)$,
\[
\varepsilon_n^2
\asymp
R^{\frac{2}{2\beta+1}}
\Big(\frac{\sigma^2}{n}\Big)^{\frac{2\beta}{2\beta+1}}.
\]

Thus, isotonic empirical Bayes shrinkage yields adaptive posterior contraction
in prediction norm, without tuning to the unknown sparsity level or smoothness.
\end{corollary}

\begin{proof}
The result follows by combining Theorem~\ref{thm:contr_sparse} (or Theorem~\ref{thm:contr_sob})
with the coordinate transformation argument of Corollary~\ref{cor:contr_pred}.
The key observation is that the prediction norm $\|X(\beta-\beta_0)\|_2$ equals the canonical
norm $\|\theta-\theta_0\|_2$ due to the orthogonality of $U$ in the SVD $X=UDV^\top$.
See Appendix~\ref{sec:contraction} for details.
\end{proof}

\section{Discussion}
\label{sec:discussion_conclusion}
This paper offers three main insights. First, many familiar regularization procedures---ridge regression, generalized ridge regression, principal components regression, and positive-part Stein shrinkage---can be viewed as special or limiting cases of empirical Bayes estimation under ordered variance structures. This perspective clarifies the role of shrinkage as a form of data-driven variance allocation rather than coefficient selection, and connects classical frequentist procedures to modern Bayesian global--local priors within a single framework. Second, isotonic empirical Bayes shrinkage adapts automatically to unknown sparsity and signal strength under monotone structure: the resulting posterior mean estimator satisfies a near-minimax risk bound (up to logarithmic factors) over ordered sparse classes, with an end-to-end guarantee available under the explicit dyadic-margin condition in Assumption~\ref{assump:margin_dyadic}. Third, the isotonic constraint stabilizes variance estimation and mitigates the collapse behavior known to arise in marginal likelihood estimation for unstructured shrinkage priors; see Section~\ref{sec:collapse_properties} for an explicit analysis.

Our theoretical formulation rests on assumptions that are standard in this literature. The Gaussian sequence reduction follows from working in an orthogonal basis (or, for regression, using singular value decomposition). The monotonicity constraint \(\tau_1^2 \ge \cdots \ge \tau_p^2\) expresses that signal strength decays along a meaningful ordering---as in polynomial regression, Sobolev-type smoothness classes, eigen-sparse (spectral) regression, and ordered basis expansions---and is far weaker than exact sparsity. Estimating hyperparameters by marginal likelihood is a classical empirical Bayes strategy dating back to \citet{robbins1956empirical} and subsequent developments \citep{morris1983parametric,jiang2009general}.

Methodologically, the proposed approach yields a convex, computationally efficient estimator based on marginal likelihood maximization under shape constraints. Unlike fully Bayesian formulations that require MCMC or variational approximations, the estimator can be computed using isotonic regression and admits a transparent interpretation in terms of shrinkage factors. The extension to regression via singular value decomposition further shows how the method adapts to effective rank and eigenvalue-aligned sparsity in general linear models. From a practical perspective, the method is particularly well suited to problems where predictors admit a natural ordering or spectral interpretation, including polynomial regression, functional data analysis, inverse problems, and high-dimensional regression with correlated designs.

The proposed framework differs from popular global--local priors such as the horseshoe or Dirichlet--Laplace in several respects. These priors assume exchangeability of local scale parameters, whereas we explicitly impose structure reflecting ordered signal strength. Our use of empirical Bayes avoids the need for hyperparameter tuning or heavy-tailed mixing distributions while still yielding sharp risk guarantees under ordered structure. More broadly, the structured viewpoint helps clarify how Bayes risk and effective degrees of freedom evolve with model complexity, offering a useful lens on phenomena such as double descent. The posterior contraction results in Section~\ref{sec:concentration_results} provide complementary Bayesian guarantees.

Several extensions merit further investigation. Other shape constraints---such as convexity or blockwise constancy---may be appropriate in certain applications and could be incorporated within the same empirical Bayes framework. Although we provide frequentist risk bounds for the posterior mean, a detailed study of uncertainty quantification and frequentist coverage of credible sets remains open, particularly in structured (ordered) settings. Finally, while we focus on Gaussian models, the underlying ideas extend naturally to generalized linear models and non-Gaussian likelihoods via variance-stabilizing transformations or local quadratic approximations. On the computational side, developing scalable implementations and approximations for larger models (including neural networks) and robust procedures for $p \gg n$ regression under structured local variance profiles remain important directions. On the shape-restriction side, it would also be interesting to develop analogous empirical Bayes global--local procedures for multidimensional or graph isotonic regression, where minimax and adaptive risk theory is now available; see, e.g., \citet{han2019isotonic,deng2020isotonic}.

\bibliography{GlobalLocal,hs-review} 

\appendix

\section{Proof of Theorem~\ref{thm:main}}\label{app:proof}

\subsection*{Proof outline}
The proof proceeds in three steps:
\begin{enumerate}
\item \textbf{Oracle risk bound (Step 1):} We first analyze the oracle shrinkage estimator
$\theta_i^{\mathrm{or}}=g_{2\lambda}(V_i)Y_i^{(-)}$ that uses the true variance profile $V_i=\theta_i^2$.
A bias--variance calculation shows that the oracle risk is $O(s\lambda\log(1+R/\lambda))$.
\item \textbf{Excess risk from variance estimation (Step 2):} The excess risk
$\mathbb{E}\|\widehat\theta-\theta^{\mathrm{or}}\|^2$ arising from estimating $V$ by
$\widehat V$ is controlled via Lemma~\ref{lem:risktransfer_main} (risk transfer under independence)
and Lemma~\ref{lem:isotonic_weighted_main} (weighted isotonic regression bound).
The key is that the cloning/splitting construction makes $\widehat V$ independent of the
Gaussian sample $Y^{(-)}$ to which shrinkage is applied.
\item \textbf{Combining bounds (Step 3):} The total risk is bounded by the sum of the oracle risk
and the excess risk, yielding the stated bound $O(s\lambda\log(ep/s)\log(1+R/\lambda))$.
\end{enumerate}
The proof is conditional on the binning-scale regularity event \textup{(B1)}; an unconditional
version follows from Proposition~\ref{prop:margin_implies_B1} under Assumption~\ref{assump:margin_dyadic}.

\subsection*{Auxiliary lemmas used in the proof}

We use the following lemmas, stated here for convenience.

\begin{lemma}[Lipschitz continuity of shrinkage map]
\label{lem:lipschitz_main}
Let $\lambda>0$ and $g(x)=x/(x+\lambda)$ for $x\ge 0$. Then
\[
|g(x)-g(y)|\le \lambda^{-1}|x-y|,\qquad \forall x,y\ge 0.
\]
\end{lemma}

\paragraph{A technical device (Gaussian cloning).}
The empirical Bayes variance estimator $\widehat{\tau}^2$ is a (nonlinear) function of $Y$, so bounds
that treat $\widehat{\tau}^2$ as independent of $Y$ are not valid as stated.
To keep the analysis transparent, we analyze an equivalent two-sample construction obtained by adding
independent Gaussian noise and splitting $Y$ into two independent ``clones.'' Specifically, let
$Y_i=\theta_i+\varepsilon_i$ with $\varepsilon_i\sim N(0,\lambda)$ and draw i.i.d.\ $Z_i\sim N(0,\lambda)$
independently of $\varepsilon$. Define
\[
Y_i^{(+)}:=Y_i+Z_i,\qquad Y_i^{(-)}:=Y_i-Z_i.
\]
Then $Y^{(+)}$ and $Y^{(-)}$ are independent and
$Y_i^{(\pm)}\sim N(\theta_i,2\lambda)$.
We estimate the variance profile from $Y^{(+)}$ and apply shrinkage to $Y^{(-)}$.

\begin{lemma}[Risk transfer under independent variance estimation]
\label{lem:risktransfer_main}
Fix $\lambda>0$ and define $g_\lambda(x)=x/(x+\lambda)$ for $x\ge 0$.
Let $Y_i\sim N(\theta_i,\lambda)$ be independent and define the deterministic target sequence
$V_i:=\theta_i^2$.
Let $\widehat{V}=(\widehat{V}_1,\dots,\widehat{V}_p)$ be any estimator independent of $Y=(Y_1,\dots,Y_p)$.
Define the (data-driven) shrinkage estimator and its oracle counterpart by
\[
\widehat{\theta}_i := g_\lambda(\widehat{V}_i)\,Y_i,
\qquad
\theta_i^{\mathrm{or}} := g_\lambda(V_i)\,Y_i.
\]
If $\widehat{V}\perp Y$, then
\begin{equation}
\label{eq:risk_transfer}
\mathbb{E}_\theta\|\widehat{\theta}-\theta^{\mathrm{or}}\|_2^2
\ \le\
\sum_{i=1}^p \mathbb{E}\Bigg[\frac{(\widehat{V}_i-V_i)^2}{V_i+\lambda}\Bigg].
\end{equation}
\end{lemma}

\begin{proof}[Proof of Lemma~\ref{lem:risktransfer_main}]
Using the identity
\[
g_\lambda(a)-g_\lambda(b)=\frac{\lambda(a-b)}{(a+\lambda)(b+\lambda)},
\]
we have for all $a,b\ge 0$ that
\[
|g_\lambda(a)-g_\lambda(b)|
=
\frac{\lambda|a-b|}{(a+\lambda)(b+\lambda)}
\le
\frac{|a-b|}{b+\lambda}.
\]
By independence of $\widehat{V}$ and $Y$,
\[
\mathbb{E}_\theta(\widehat{\theta}_i-\theta_i^{\mathrm{or}})^2
=
\mathbb{E}\big(g_\lambda(\widehat{V}_i)-g_\lambda(V_i)\big)^2\,\mathbb{E}_\theta Y_i^2
\le
\mathbb{E}\Big[\frac{(\widehat{V}_i-V_i)^2}{(V_i+\lambda)^2}\Big]\,(V_i+\lambda),
\]
where $\mathbb{E}_\theta Y_i^2=\theta_i^2+\lambda=V_i+\lambda$.
Summing over $i$ yields \eqref{eq:risk_transfer}.
\end{proof}

\paragraph{Orlicz norms and tail classes.}
For a random variable $X$, define the Orlicz norms
\[
\|X\|_{\psi_1}:=\inf\{c>0:\ \mathbb{E}\exp(|X|/c)\le 2\},
\qquad
\|X\|_{\psi_2}:=\inf\{c>0:\ \mathbb{E}\exp(X^2/c^2)\le 2\}.
\]
We say $X$ is sub-exponential if $\|X\|_{\psi_1}<\infty$ and sub-Gaussian if $\|X\|_{\psi_2}<\infty$.

\begin{proposition}[Margin implies correct dyadic binning with high probability]
\label{prop:margin_implies_B1}
In the setting of Lemma~\ref{lem:isotonic_weighted_main}, assume Assumption~\ref{assump:margin_dyadic}.
There exist universal constants $c_1,c_2>0$ such that for all $p\ge 3$,
\[
\sup_{\theta\in\Theta_{\downarrow}(s,R)}\mathbb{P}_\theta\big((B1)^c\big)
\ \le\
c_1\,p^2\exp\!\Big(-c_2\,\kappa^2\frac{\nu}{R+\nu}\Big).
\]
\noindent In particular, if \(\kappa^2\nu/(R+\nu)\ge C\log p\) for a sufficiently large universal constant \(C>0\),
then the right-hand side is at most \(c\,p^{-3}\) for a universal constant \(c>0\).
\end{proposition}

\begin{proof}
Let $\widetilde V=\Pi_{\mathcal{C}}(X^{(1)})$ be the pilot isotonic fit used to form the bins in Lemma~\ref{lem:isotonic_weighted_main},
where $\mathcal{C}=\{v\in\mathbb{R}^p:\ v_1\ge\cdots\ge v_p\ge 0\}$ and $X^{(1)}=V+\xi$ with independent mean-zero
sub-exponential coordinates.
By the max--min characterization of isotonic regression (see, e.g., \citet[Ch.~1]{robertson1988order} or
\citet{barlow1972order}), for each $i$,
\[
\widetilde V_i
=
\Big(\min_{1\le h\le i}\ \max_{i\le j\le p}\ \bar X^{(1)}_{h:j}\Big)_+,
\qquad
\bar X^{(1)}_{h:j}:=\frac{1}{j-h+1}\sum_{\ell=h}^j X^{(1)}_\ell.
\]
Since $V$ is nonincreasing, the same representation applied to $V$ gives $V_i=(\min_{h\le i}\max_{j\ge i}\bar V_{h:j})_+$
where $\bar V_{h:j}$ is the average of $V$ over $[h,j]$. Using the inequality \(|(a)_+-(b)_+|\le |a-b|\) and the
Lipschitz property of $\min$ and $\max$ operations, we obtain the deterministic bound
\[
|\widetilde V_i-V_i|
\le
\max_{1\le h\le i\le j\le p}\ |\bar\xi_{h:j}|,
\qquad
\bar\xi_{h:j}:=\frac{1}{j-h+1}\sum_{\ell=h}^j \xi_\ell.
\]
(The factor of 2 in some presentations arises from chaining $\min$-$\max$ bounds; we use the tighter form here.)
Therefore, on the event
\[
\mathcal{E}:=\Big\{\max_{1\le h\le j\le p}|\bar\xi_{h:j}|\le (\kappa/4)\,\nu\Big\},
\]
we have $|\widetilde V_i-V_i|\le (\kappa/4)\nu$ for all $i$.
Combining this with the margin condition \eqref{eq:margin_dyadic} shows that $\widetilde V_i+\nu$ falls in the same dyadic
interval $[t_{m^\star(i)},t_{m^\star(i)+1})$ as $V_i+\nu$ for all $i$, hence the induced bins satisfy \textup{(B1)}.
It remains to bound $\mathbb{P}(\mathcal{E}^c)$.

For the squared-Gaussian proxy noise $\xi_i=(Y_i^{(1)})^2-\nu-\theta_i^2$, we have
$\|\xi_i\|_{\psi_1}\lesssim \nu+\sqrt{\nu V_i}\le \nu+\sqrt{\nu R}$ uniformly, hence by Bernstein's inequality,
for any fixed block $[h,j]$ of length $r=j-h+1$ and any $t>0$,
\[
\mathbb{P}\!\left(|\bar\xi_{h:j}|>t\right)
\le
2\exp\!\Big(-c_0\,r\min\Big\{\frac{t^2}{\Lambda^2},\frac{t}{\Lambda}\Big\}\Big),
\qquad
\Lambda:=\nu+\sqrt{\nu R}.
\]
Taking $t=(\kappa/4)\nu$ and using $\nu\le \Lambda$ yields
$\mathbb{P}(|\bar\xi_{h:j}|>t)\le 2\exp(-c_1 r \kappa^2 \nu^2/\Lambda^2)\le 2\exp(-c_2 r \kappa^2 \nu/(R+\nu))$.
There are at most $p^2$ blocks, so a union bound gives
\[
\mathbb{P}(\mathcal{E}^c)
\le
2p^2\exp\!\Big(-c_2\,\kappa^2\frac{\nu}{R+\nu}\Big).
\]
This proves the displayed exponential bound. The polynomial bound \(c\,p^{-3}\) follows under the additional scaling
\(\kappa^2\nu/(R+\nu)\gtrsim \log p\) by choosing the implicit constant large enough.
\end{proof}


\begin{lemma}[Isotonic EB variance estimation risk bound]
\label{lem:isotonic_main}
Let $X_i=V_i+\xi_i$ where $V_1\ge \cdots \ge V_p\ge 0$ is a nonincreasing sequence and $\xi_i$ are independent,
mean-zero, sub-exponential errors with scale parameter $\lambda>0$ (i.e., $\|\xi_i\|_{\psi_1}\lesssim \lambda$
uniformly in $i$).
Let $\widehat{V}$ be the isotonic least squares estimator of $V$ (PAVA) under the constraint
$\widehat{V}_1\ge\cdots\ge\widehat{V}_p\ge 0$.
Then there exists a universal constant $C_0>0$ such that for every $1\le s\le p$,
\begin{equation}
\label{eq:isotonic_oracle}
\mathbb{E}\sum_{i=1}^p(\widehat{V}_i-V_i)^2
\ \le\
C_0\,\lambda^2\, s\log\!\Big(\frac{ep}{s}\Big)
\ +\
C_0\,\inf_{\substack{u\in\mathbb{R}^p:\\ u_1\ge\cdots\ge u_p\ge 0\\ \#\{i:u_i\neq u_{i+1}\}\le s}}
\sum_{i=1}^p (u_i-V_i)^2.
\end{equation}
\end{lemma}

\begin{lemma}[Block projection risk bound for the monotone cone (Zhang, 2002)]
\label{lem:block_proj_zhang}
Let $m\ge 1$ and let $\mathcal{C}_m:=\{v\in\mathbb{R}^m:\ v_1\ge v_2\ge\cdots\ge v_m\}$ be the nonincreasing
monotone cone.
Let $\varepsilon=(\varepsilon_1,\dots,\varepsilon_m)$ have independent coordinates with
$\mathbb{E}\varepsilon_i=0$, $\mathbb{E}\varepsilon_i^2\le \sigma^2$ for all $i$, and assume in addition that the coordinates are sub-exponential with a uniform scale control (e.g.\ $\|\varepsilon_i\|_{\psi_1}\lesssim \sigma$ uniformly in $i$).
Then there exists a universal constant $c_0>0$ such that
\[
\mathbb{E}\big\|\Pi_{\mathcal{C}_m}(\varepsilon)\big\|_2^2 \ \le\ c_0\,\sigma^2\,\log(em).
\]
\end{lemma}

\begin{proof}[Proof sketch]
Let $t_i=i$ and consider the isotonic regression model $y_i=f(t_i)+\varepsilon_i$ for $i=1,\dots,m$ with the
constant nondecreasing function $f\equiv 0$.
Let $\widehat f$ be the least-squares isotonic fit under the nondecreasing constraint. Then
$(\widehat f(t_1),\dots,\widehat f(t_m))=\Pi_{\widetilde{\mathcal{C}}_m}(y)$ where
$\widetilde{\mathcal{C}}_m:=\{u\in\mathbb{R}^m:\ u_1\le\cdots\le u_m\}$.
By symmetry, $\Pi_{\widetilde{\mathcal{C}}_m}(y)=-\Pi_{\mathcal{C}_m}(-y)$, and since $y=\varepsilon$,
\(\|\Pi_{\widetilde{\mathcal{C}}_m}(\varepsilon)\|_2=\|\Pi_{\mathcal{C}_m}(\varepsilon)\|_2\).

Under sub-exponential errors with $\|\varepsilon_i\|_{\psi_1}\lesssim \sigma$, the isotonic least squares risk for the constant signal satisfies a logarithmic bound of order $\sigma^2\log(em)/m$ (the Gaussian case is established in \citet[Theorem~2.2(ii)]{zhang2002risk} and the extension to sub-exponential errors follows by a standard truncation/symmetrization argument). Multiplying by $m$ yields the stated projection-norm bound.
\end{proof}

\begin{proof}[Proof of Lemma~\ref{lem:isotonic_main}]
We use the standard projection-onto-a-cone oracle-inequality argument for isotonic regression on the monotone cone,
combined with a bound on the expected squared norm of noise projected onto blockwise monotone cones.
Related sharp oracle inequalities for isotonic least squares (including adaptive bounds for piecewise-constant isotonic
signals and misspecification) can be found in \citet{bellec2018sharp,gao2020estimation}; see also the survey
\citet{guntuboyina2018nonparametric}.

\paragraph{Step 1: deterministic reduction to a tangent-cone noise projection and an $s$-block bound.}
Let $\mathcal{C}:=\{v\in\mathbb{R}^p:\ v_1\ge\cdots\ge v_p\ge 0\}$ and note that the isotonic estimator is the Euclidean
projection $\widehat V=\Pi_{\mathcal{C}}(X)$ with $X=V+\xi$.
The projection optimality condition gives, for all $v\in \mathcal{C}$,
\[
\langle X-\widehat V,\ v-\widehat V\rangle \le 0.
\]
Fix any comparator $u\in \mathcal{C}$ and set $h:=\widehat V-u$. Taking $v=u$ and expanding yields the deterministic inequality
\begin{equation}
\label{eq:proj_oracle_det}
\|\widehat V-V\|_2^2
\le
\|u-V\|_2^2
\;+\;
\Big(2\langle \xi,\ h\rangle-\|h\|_2^2\Big).
\end{equation}
Since $u\in \mathcal{C}$ and $\widehat V\in \mathcal{C}$, we have $h=\widehat V-u\in \mathcal{C}-u\subseteq T_{\mathcal{C}}(u)$, where
\[
T_{\mathcal{C}}(u):=\overline{\{t(v-u):\ t\ge 0,\ v\in \mathcal{C}\}}
\]
is the tangent cone of $\mathcal{C}$ at $u$.
Therefore
\[
2\langle \xi,\ h\rangle-\|h\|_2^2
\le
\sup_{g\in T_{\mathcal{C}}(u)}\Big(2\langle \xi,\ g\rangle-\|g\|_2^2\Big)
=
\|\Pi_{T_{\mathcal{C}}(u)}(\xi)\|_2^2,
\]
where the last equality is the standard variational identity for projections onto a closed convex cone.
Taking expectations in \eqref{eq:proj_oracle_det} gives
\begin{equation}
\label{eq:proj_oracle_exp}
\mathbb{E}\|\widehat V-V\|_2^2
\le
\|u-V\|_2^2
\;+\;
\mathbb{E}\|\Pi_{T_{\mathcal{C}}(u)}(\xi)\|_2^2.
\end{equation}

\paragraph{Tangent-cone decomposition for an $s$-block $u$.}
Let $u\in \mathcal{C}$ have $k\le s$ jumps, i.e. it is constant on $k+1$ contiguous blocks $B_1,\dots,B_{k+1}$ with sizes
$m_b:=|B_b|$ and strict decreases across block boundaries.
For the monotone cone, the only active inequalities at $u$ occur within the constant blocks: within a block
$B_b$ we have $u_i-u_{i+1}=0$, while at a block boundary $u_i-u_{i+1}>0$ and the corresponding inequality is inactive.
Equivalently, first-order feasible perturbations $g$ must preserve monotonicity inside each constant block but are not
constrained across strict boundaries. Hence
$T_{\mathcal{C}}(u)$ decomposes as a Cartesian product of monotone cones on the blocks:
\[
T_{\mathcal{C}}(u)
=
\big\{g\in\mathbb{R}^p:\ g|_{B_b}\in \mathcal{C}_{m_b}\ \text{for each }b\big\},
\qquad
\mathcal{C}_m:=\{w\in\mathbb{R}^m:\ w_1\ge\cdots\ge w_m\}.
\]
Here we have dropped the nonnegativity constraint within each block (replacing the tangent cone of $v_1\ge\cdots\ge v_m\ge 0$ by the larger cone $\mathcal{C}_m$), which can only increase the projection norm and therefore preserves the required upper bound.
Consequently the projection splits blockwise and
\[
\|\Pi_{T_{\mathcal{C}}(u)}(\xi)\|_2^2
=
\sum_{b=1}^{k+1}\ \|\Pi_{\mathcal{C}_{m_b}}(\xi_{B_b})\|_2^2,
\qquad
\mathbb{E}\|\Pi_{T_{\mathcal{C}}(u)}(\xi)\|_2^2
=
\sum_{b=1}^{k+1}\ \mathbb{E}\|\Pi_{\mathcal{C}_{m_b}}(\xi_{B_b})\|_2^2.
\]

\paragraph{Bounding the blockwise projection term.}
Since $\|\xi_i\|_{\psi_1}\lesssim \lambda$ implies finite exponential moments and in particular bounded variances
$\mathrm{Var}(\xi_i)\lesssim \lambda^2$, the isotonic least squares estimator over a constant mean (the case
$V\equiv 0$ on a block) satisfies the classical logarithmic risk bound.
Concretely, applying Lemma~\ref{lem:block_proj_zhang} on a block of length $m$ and using $\mathbb{E}\xi_i^2\lesssim \lambda^2$ gives
\[
\mathbb{E}\|\Pi_{\mathcal{C}_m}(\xi_{1:m})\|_2^2
\ \lesssim\ \lambda^2 \log(em).
\]
Therefore,
\[
\mathbb{E}\|\Pi_{T_{\mathcal{C}}(u)}(\xi)\|_2^2
\ \lesssim\
\lambda^2\sum_{b=1}^{k+1}\log(em_b)
\ \le\
\lambda^2 (k+1)\log\!\Big(\frac{ep}{k+1}\Big)
\ \lesssim\
\lambda^2 s\log\!\Big(\frac{ep}{s}\Big),
\]
where we used $\sum_b m_b=p$ and concavity of $\log$.
Plugging this into \eqref{eq:proj_oracle_exp} and taking the infimum over admissible $u$ yields
\eqref{eq:isotonic_oracle}.
\end{proof}

\begin{lemma}[Weighted isotonic regression as an unweighted projection]
\label{lem:weighted_iso_projection}
Let $w_1,\dots,w_p>0$ and define the weighted inner product and norm
\[
\langle a,b\rangle_w := \sum_{i=1}^p w_i a_i b_i,
\qquad
\|a\|_w^2 := \langle a,a\rangle_w.
\]
Let $\mathcal{C}:=\{v\in\mathbb{R}^p: v_1\ge \cdots \ge v_p\ge 0\}$ be the (nonincreasing, nonnegative) monotone cone.
For any $x\in\mathbb{R}^p$, let
\[
\widehat v \in \arg\min_{v\in \mathcal{C}}\ \|x-v\|_w^2
\]
be the weighted isotonic least squares fit.
Let $W:=\mathrm{diag}(\sqrt{w_1},\dots,\sqrt{w_p})$ and let $\Pi_D(\cdot)$ denote the Euclidean projection onto a
closed convex set $D\subseteq\mathbb{R}^p$.
Then
\[
W\widehat v \;=\; \Pi_{W \mathcal{C}}(W x)
\qquad\text{and hence}\qquad
\widehat v \;=\; W^{-1}\Pi_{W \mathcal{C}}(W x).
\]
\end{lemma}

\begin{proof}
For any $v\in\mathbb{R}^p$ we have $\|x-v\|_w^2=\|W(x-v)\|_2^2$.
Since $W$ is an invertible linear map, minimizing $\|W(x-v)\|_2^2$ over $v\in C$ is equivalent to minimizing
$\|Wx-u\|_2^2$ over $u\in W \mathcal{C}$, with the one-to-one change of variables $u=Wv$.
Thus $u^\star:=W\widehat v$ is the Euclidean projection of $Wx$ onto the closed convex set $W\mathcal{C}$.
\end{proof}

\begin{lemma}[Refinement count]
\label{lem:refinement_count}
Let $u\in\mathbb{R}^p$ be piecewise-constant with at most $s$ jumps (so at most $s+1$ constant blocks), and let
$(I_m)_{m=0}^{M-1}$ be a partition of $\{1,\dots,p\}$ into $M$ (possibly empty) contiguous bins.
Form the refinement by intersecting each constant block of $u$ with each bin $I_m$.
Then the refinement produces at most $(s+1)M$ (possibly empty) sub-blocks.
\end{lemma}
\begin{proof}
Each constant block of $u$ is a contiguous interval. Intersecting a fixed interval with the $M$ disjoint bins produces
at most $M$ (possibly empty) sub-intervals. Summing over the at most $s+1$ constant blocks gives the bound.
\end{proof}

\begin{lemma}[Weighted isotonic risk bound for the squared-Gaussian variance proxy]
\label{lem:isotonic_weighted_main}
Let $\nu>0$ and assume $Y_i\sim N(\theta_i,\nu)$ independently.
Define \(V_i:=\theta_i^2\) and suppose \(V_1\ge\cdots\ge V_p\ge 0\) with \(V_1\le R\).

\paragraph{A cross-fit variance-profile estimator.}
Draw i.i.d.\ auxiliary noises \(Z_i\sim N(0,\nu)\), independent of \((Y_i)\), and form the Gaussian split
\[
Y_i^{(1)}:=\frac{Y_i+Z_i}{\sqrt 2},\qquad
Y_i^{(2)}:=\frac{Y_i-Z_i}{\sqrt 2},
\]
so that \(\big(Y_i^{(1)}\big)_{i\le p}\) and \(\big(Y_i^{(2)}\big)_{i\le p}\) are independent and each satisfies
\(Y_i^{(a)}\sim N(\theta_i,\nu)\).
Define variance proxies
\[
X_i^{(a)} := \big(Y_i^{(a)}\big)^2-\nu,\qquad a\in\{1,2\},
\]
so \(\mathbb{E}_\theta[X_i^{(a)}]=V_i\).

Compute a coarse pilot \(\widetilde V\) from fold \(a=1\) by unweighted isotonic regression:
\[
\widetilde V
\in
\arg\min_{v_1\ge\cdots\ge v_p\ge 0}\ \sum_{i=1}^p \big(X_i^{(1)}-v_i\big)^2,
\]
and cap it by \(\widetilde V^{\,\cap}_i := \min\{\widetilde V_i,\,R\}\).
Let \(M:=1+\left\lceil \log_2\!\left(1+\frac{R}{\nu}\right)\right\rceil\) and define dyadic thresholds
\(t_m:=2^m\nu\) for \(m=0,1,\dots,M\) (so \(t_0=\nu\) and \(t_M> R+\nu\)).
Using the monotonicity of \(\widetilde V^{\,\cap}\), define the (possibly empty) contiguous index blocks
\[
I_m:=\left\{i\le p:\ t_m \le \widetilde V^{\,\cap}_i+\nu < t_{m+1}\right\},
\qquad m=0,1,\dots,M-1,
\]
which form a partition of \(\{1,\dots,p\}\) (some blocks may be empty).
Assume that these pilot-based bins satisfy the binning-scale regularity condition (Assumption~\ref{assump:B1}), i.e.
\[
\text{for all } m\in\{0,\dots,M-1\}\text{ and all } i\in I_m,\qquad
t_m \le V_i+\nu < t_{m+1}=2t_m.
\tag{B1}
\]
On fold \(a=2\), define the raw weighted isotonic least squares fit \(\widehat V^{\mathrm{raw}}\) by
\[
\widehat V^{\mathrm{raw}}
\in
\arg\min_{v_1\ge\cdots\ge v_p\ge 0}\ \sum_{m=0}^{M-1}\ \sum_{i\in I_m}\ \frac{1}{t_m}\,\big(X_i^{(2)}-v_i\big)^2.
\]
Define the final variance-profile estimator by truncation at the known upper bound \(R\)
\[
\widehat V_i := \min\{\widehat V_i^{\mathrm{raw}},\,R\},\qquad i=1,\dots,p.
\]
Then there exists a universal constant \(C>0\) such that for every \(1\le s\le p\),
\begin{equation}
\label{eq:isotonic_weighted}
\mathbb{E}_\theta\sum_{i=1}^p \frac{(\widehat{V}_i-V_i)^2}{V_i+\nu}
\ \le\
C\,\nu\, s\log\!\Big(\frac{ep}{s}\Big)\,
\log\!\Big(1+\frac{R}{\nu}\Big)
\ +\
C\,\inf_{\substack{u\in\mathbb{R}^p:\\ u_1\ge\cdots\ge u_p\ge 0\\ \#\{i:u_i\neq u_{i+1}\}\le s}}
\sum_{i=1}^p \frac{(u_i-V_i)^2}{V_i+\nu}.
\end{equation}
\end{lemma}

\begin{proof}
We work conditional on the pilot/binning sigma-field so that the weights are deterministic, and we use only
sub-exponential concentration (Bernstein-type bounds) appropriate for the squared-Gaussian proxy noise.

\paragraph{Step 0: second-fold noise decomposition and tail scale.}
The Gaussian split ensures that \(Y_i^{(2)}\sim N(\theta_i,\nu)\), so we may write
\[
Y_i^{(2)}=\theta_i+\varepsilon_i^{(2)},\qquad \varepsilon_i^{(2)}\sim N(0,\nu),
\]
independently across \(i\). Hence \(X_i^{(2)}=(Y_i^{(2)})^2-\nu\) satisfies
\[
X_i^{(2)} = V_i + \xi_i^{(2)},
\qquad
\xi_i^{(2)} := 2\theta_i\varepsilon_i^{(2)} + \big((\varepsilon_i^{(2)})^2-\nu\big),
\]
with \(\xi_i^{(2)}\) independent, mean zero. Moreover,
\[
\mathrm{Var}(\xi_i^{(2)})=2\nu^2+4\nu V_i \ \lesssim\ \nu(V_i+\nu),
\qquad
\|\xi_i^{(2)}\|_{\psi_1}\ \lesssim\ \nu + \sqrt{\nu V_i}\ \le\ \nu + \sqrt{\nu R}.
\]
(The \(\psi_1\) bound follows from sub-exponentiality of \((\varepsilon_i^{(2)})^2-\nu\) and the product
term \(2\theta_i\varepsilon_i^{(2)}\), using \(V_i\le R\).)

More explicitly: if \(\varepsilon\sim N(0,\nu)\), then \((\varepsilon^2-\nu)\) is sub-exponential with
\(\|\,\varepsilon^2-\nu\,\|_{\psi_1}\lesssim \nu\). Also, \(2\theta_i\varepsilon_i^{(2)}\) is sub-Gaussian with
\(\|2\theta_i\varepsilon_i^{(2)}\|_{\psi_2}\lesssim |\theta_i|\sqrt{\nu}\), hence sub-exponential with
\(\|2\theta_i\varepsilon_i^{(2)}\|_{\psi_1}\lesssim |\theta_i|\sqrt{\nu}=\sqrt{\nu V_i}\). By the triangle inequality
for Orlicz norms, \(\|\xi_i^{(2)}\|_{\psi_1}\le \|2\theta_i\varepsilon_i^{(2)}\|_{\psi_1}+\|(\varepsilon_i^{(2)})^2-\nu\|_{\psi_1}
\lesssim \nu+\sqrt{\nu V_i}\).

\paragraph{Step 1: set up weighted isotonic regression in the correct geometry.}
Let \(\mathcal{F}_1\) denote the sigma-field generated by the pilot/binning step. Conditional on \(\mathcal{F}_1\),
the partition \((I_m)\) and the weights \(w_i:=t_{m(i)}^{-1}\) are deterministic. By \textup{(B1)}, for every \(i\),
\[
\frac{1}{V_i+\nu}\ \le\ w_i\ \le\ \frac{2}{V_i+\nu}.
\]
Write \(\mathcal{C}:=\{v\in\mathbb{R}^p:\ v_1\ge\cdots\ge v_p\ge 0\}\). The raw estimator \(\widehat V^{\mathrm{raw}}\) is the
weighted projection of \(X^{(2)}\) onto \(\mathcal{C}\) in the weighted norm \(\|\cdot\|_w^2=\sum_i w_i(\cdot)_i^2\).
Equivalently, by Lemma~\ref{lem:weighted_iso_projection}, with \(W:=\mathrm{diag}(\sqrt{w_1},\dots,\sqrt{w_p})\),
\[
W\widehat V^{\mathrm{raw}}=\Pi_{W\mathcal{C}}(W X^{(2)}).
\]
In particular,
\[
\sum_{i=1}^p w_i\big(\widehat V^{\mathrm{raw}}_i-V_i\big)^2
=
\big\|W\widehat V^{\mathrm{raw}}-WV\big\|_2^2.
\]
Finally, since \(V_i\le R\) for all \(i\), the truncation \(\widehat V_i=\min\{\widehat V_i^{\mathrm{raw}},R\}\)
can only decrease the squared error coordinatewise, hence also in the weighted loss:
for every \(i\), \((\widehat V_i-V_i)^2\le(\widehat V_i^{\mathrm{raw}}-V_i)^2\), so
\(\sum_i w_i(\widehat V_i-V_i)^2\le \sum_i w_i(\widehat V_i^{\mathrm{raw}}-V_i)^2\).
Therefore it suffices to control \(\widehat V^{\mathrm{raw}}\).

\paragraph{Step 2: a noise normalization that yields uniformly bounded variance.}
Recall \(X_i^{(2)}=V_i+\xi_i^{(2)}\) with \(\xi_i^{(2)}\) independent, mean zero, and
\(\mathrm{Var}(\xi_i^{(2)})\lesssim \nu(V_i+\nu)\). Define the normalized noise \(\eta:=W\xi^{(2)}\), i.e.
\(\eta_i=\sqrt{w_i}\,\xi_i^{(2)}\). Under \textup{(B1)},
\[
\mathrm{Var}(\eta_i)\ =\ w_i\,\mathrm{Var}(\xi_i^{(2)})\ \lesssim\ \nu,
\]
uniformly in \(i\). The sub-exponential norm satisfies \(\|\xi_i^{(2)}\|_{\psi_1}\lesssim \nu+\sqrt{\nu R}\) and
\(w_i\le 2/(V_i+\nu)\le 2/\nu\), so \(\|\eta_i\|_{\psi_1}\lesssim \sqrt{2/\nu}\cdot(\nu+\sqrt{\nu R})\lesssim \sqrt{\nu}(1+\sqrt{R/\nu})\), which is bounded (though depending on \(R/\nu\)).
Crucially, the variance bound \(\mathrm{Var}(\eta_i)\lesssim \nu\) together with the sub-exponential tail control
are what we use to justify applying Lemma~\ref{lem:block_proj_zhang} on each refined block below.

\paragraph{Step 3: an oracle inequality for the weighted projection via a bin-refinement argument.}
Write \(Y:=W X^{(2)}=WV+\eta\) and note \(WV\in W\mathcal{C}\). Since \(W\widehat V^{\mathrm{raw}}=\Pi_{W\mathcal{C}}(Y)\),
the same projection argument as in Lemma~\ref{lem:isotonic_main} (but in Euclidean norm on the transformed space)
gives, for any comparator \(u\in \mathcal{C}\),
\[
\|W\widehat V^{\mathrm{raw}}-WV\|_2^2
\le
\|Wu-WV\|_2^2
\;+\;
\|\Pi_{T_{W\mathcal{C}}(Wu)}(\eta)\|_2^2,
\]
where \(T_{W\mathcal{C}}(Wu)\) is the tangent cone of \(W\mathcal{C}\) at \(Wu\). Because \(W\) is invertible and linear,
\(T_{W\mathcal{C}}(Wu)=W T_{\mathcal{C}}(u)\).

To bound the noise-projection term, we upper bound the cone \(WT_{\mathcal{C}}(u)\) by a larger cone that decouples along the
dyadic bins. Let \(u\) have at most \(s\) jumps, so it is constant on at most \(s+1\) contiguous blocks.
Refine each constant block by intersecting it with the dyadic bins \((I_m)_{m=0}^{M-1}\). Since both the constant blocks
and the dyadic bins form partitions of \(\{1,\ldots,p\}\), the number of non-empty refined sub-blocks is at most
\(s+M\) (each new boundary comes from either a jump in \(u\) or a dyadic threshold, giving at most \(s+(M-1)+1=s+M\)
boundary points, hence at most \(s+M\) sub-blocks). We use the looser bound \(B\le (s+1)M\) which suffices for
our purposes. Enlarging the cone by dropping the
monotonicity constraints across sub-block boundaries yields an (orthogonal) product cone, hence
\[
 \|\Pi_{W T_{\mathcal{C}}(u)}(\eta)\|_2^2
\ \le\
 \sum_{b=1}^{B}\ \|\Pi_{\mathcal{C}_{m_b}}(\eta_{B_b})\|_2^2,
\qquad B\le (s+1)M,
\]
 where \(B_b\) are the refined sub-blocks, \(m_b:=|B_b|\), and \(\mathcal{C}_{m_b}\) is the monotone cone on the sub-block
(nonincreasing sequences in \(\mathbb{R}^{m_b}\)).
Here we used that the weights \(w_i\) are constant on each dyadic bin \(I_m\), hence constant on each refined sub-block
\(B_b\subseteq I_m\); therefore, the restriction of \(W\) to \(B_b\) is a positive scalar multiple of the identity, and
the corresponding scaled monotone cone is the same cone \(\mathcal{C}_{m_b}\) (as a cone).
On each sub-block, the entries of \(\eta_{B_b}\) are independent, mean zero, and satisfy
\(\mathrm{Var}(\eta_i)\lesssim \nu\) and \(\|\eta_i\|_{\psi_1}\lesssim \sqrt{\nu}(1+\sqrt{R/\nu})\).
 Applying Lemma~\ref{lem:block_proj_zhang} on each sub-block (with \(\sigma^2\asymp \nu\)) gives
 \(\mathbb{E}\|\Pi_{\mathcal{C}_{m_b}}(\eta_{B_b})\|_2^2 \lesssim \nu\log(e m_b)\).
Summing over \(b\) and using \(\sum_b m_b=p\) and concavity of \(\log\) yields
\[
\mathbb{E}\|\Pi_{W T_{\mathcal{C}}(u)}(\eta)\|_2^2
\ \lesssim\
\nu\,B\,\log\!\Big(\frac{ep}{B}\Big)
\ \lesssim\
\nu\, s\,\log\!\Big(\frac{ep}{s}\Big)\,\log\!\Big(1+\frac{R}{\nu}\Big),
\]
where the last step uses \(B\le (s+1)M\) and the elementary bound
\[
B\log\!\Big(\frac{ep}{B}\Big)
\le
\max\Big\{s\log\!\Big(\frac{ep}{s}\Big),\ (s+1)M\log\!\Big(\frac{ep}{s}\Big)\Big\}
\lesssim
s\,\log\!\Big(\frac{ep}{s}\Big)\,M,
\]
since for any \(a\ge 1\),
\((as)\log\!\big(ep/(as)\big)=as\big(\log(ep/s)-\log a\big)\le as\log(ep/s)\), and \(M\asymp \log(1+R/\nu)\).

\paragraph{Step 4: return to the self-normalized loss and include approximation.}
Under \textup{(B1)}, we have \(w_i\asymp 1/(V_i+\nu)\) within a factor \(2\), hence
\[
\sum_{i=1}^p \frac{(\widehat V_i-V_i)^2}{V_i+\nu}
\ \le\
\sum_{i=1}^p w_i(\widehat V_i-V_i)^2.
\]
Likewise, for any \(u\in C\),
\(\sum_i w_i(u_i-V_i)^2 \le 2\sum_i (u_i-V_i)^2/(V_i+\nu)\),
since under \textup{(B1)} we have \(w_i\le 2/(V_i+\nu)\).
Combining the oracle inequality with the bound on the noise-projection term yields \eqref{eq:isotonic_weighted}.
\end{proof}

\begin{lemma}[Minimax lower bound for ordered sparse class]
\label{lem:lower_main}
Assume $\sigma^2/n \le R$. There exists a universal constant $c>0$ such that
\[
\inf_{\tilde{\theta}}
\sup_{\theta\in\Theta_{\downarrow}(s,R)}
\mathbb{E}_\theta\|\tilde{\theta}-\theta\|_2^2
\ \ge\
c\, s\frac{\sigma^2}{n}.
\]
\end{lemma}

\begin{proof}
Let $\lambda=\sigma^2/n$ and set $a=\sqrt{\lambda}\le \sqrt{R}$.
Consider the subset of parameters
\[
\Theta^\pm := \left\{\theta\in\mathbb{R}^p:\ \theta_i\in\{+a,-a\}\ (i=1,\dots,s),\ \theta_i=0\ (i>s)\right\}.
\]
Then $\Theta^\pm\subset\Theta_{\downarrow}(s,R)$ since $\theta_1^2=\cdots=\theta_s^2=a^2\le R$ and $\theta_i=0$ for $i>s$.

For each $u\in\{\pm1\}^s$, write $\theta^{(u)}_i=u_i a$ for $i\le s$ and $\theta^{(u)}_i=0$ for $i>s$,
and let $P_u$ denote the joint law of $Y$ under $\theta=\theta^{(u)}$ in model \eqref{eq:sequence}.
For two vectors $u,u'$ that differ in exactly one coordinate, say $u_j\neq u'_j$ and $u_i=u'_i$ for $i\neq j$,
the likelihoods differ only in coordinate $j$, and
\[
\mathrm{TV}(P_u,P_{u'})=\mathrm{TV}\!\left(N(a,\lambda),N(-a,\lambda)\right)
=2\Phi\!\left(\frac{a}{\sqrt{\lambda}}\right)-1
=1-2\Phi\!\left(-\frac{a}{\sqrt{\lambda}}\right)=1-2\Phi(-1),
\]
where $\Phi$ is the standard normal cdf.
By Assouad's lemma (applied to the hypercube $\{\pm1\}^s$ with squared $\ell_2$ loss),
\[
\inf_{\tilde{\theta}}\ \sup_{u\in\{\pm1\}^s}\ \mathbb{E}_{u}\|\tilde{\theta}-\theta^{(u)}\|_2^2
\ \ge\
\frac{s a^2}{4}\left(1-\mathrm{TV}(P_u,P_{u'})\right)
\ =\
\frac{s\lambda}{4}\left(2\Phi(-1)\right).
\]
Since $\Theta^\pm\subset\Theta_{\downarrow}(s,R)$, the same bound holds for the supremum over $\Theta_{\downarrow}(s,R)$.
This gives the claim with $c=\Phi(-1)/2$.
\end{proof}

\subsection*{Proof of Theorem~\ref{thm:main} (known $\sigma^2$)}

\begin{proof}
Fix $\theta\in\Theta_{\downarrow}(s,R)$ and write $\lambda=\sigma^2/n$.
We use the Gaussian cloning device described above to create two independent samples
$Y^{(+)}$ and $Y^{(-)}$, each distributed as $N(\theta,2\lambda I_p)$.
Define the (deterministic) variance proxy sequence $V_i:=\theta_i^2$ and note that
$V_1\ge\cdots\ge V_p\ge 0$ under $\Theta_{\downarrow}(s,R)$.

Construct the cross-fit weighted isotonic estimator $\widehat V$ from $Y^{(+)}$ as in
Lemma~\ref{lem:isotonic_weighted_main} with $\nu=2\lambda$, so that $\widehat V$ is independent of the
independent clone $Y^{(-)}$ to which shrinkage is applied.
Finally define the shrinkage estimator (applied to the independent clone $Y^{(-)}$)
\[
\widehat{\theta}_i := g_{2\lambda}(\widehat{V}_i)\,Y_i^{(-)},
\qquad g_{2\lambda}(x)=\frac{x}{x+2\lambda}.
\]
Define the oracle shrinkage estimator (same form with $V_i$ in place of $\widehat{V}_i$) by
\[
\theta_i^{\mathrm{or}}:= g_{2\lambda}(V_i)\,Y_i^{(-)}.
\]
We decompose the risk as
\begin{equation}
\label{eq:decomp_main}
\mathbb{E}_\theta\|\widehat{\theta}-\theta\|_2^2
\ \le\
2\,\mathbb{E}_\theta\|\widehat{\theta}-\theta^{\mathrm{or}}\|_2^2
+
2\,\mathbb{E}_\theta\|\theta^{\mathrm{or}}-\theta\|_2^2.
\end{equation}

\paragraph{Step 1: Oracle risk bound.}
For each $i$, a standard bias--variance calculation with noise variance $2\lambda$ gives
\[
\mathbb{E}_\theta\big(\theta_i^{\mathrm{or}}-\theta_i\big)^2
=
\underbrace{2\lambda\Big(\frac{V_i}{V_i+2\lambda}\Big)^2}_{\text{variance}}
+
\underbrace{\theta_i^2\Big(\frac{2\lambda}{V_i+2\lambda}\Big)^2}_{\text{bias}^2}
\ =\
2\lambda\,\frac{V_i}{V_i+2\lambda}.
\]
Summing over $i$ yields
\begin{equation}
\label{eq:oracle_main}
\mathbb{E}_\theta\|\theta^{\mathrm{or}}-\theta\|_2^2
\ \le\
\sum_{i=1}^p 2\lambda\,\frac{V_i}{V_i+2\lambda}.
\end{equation}
Now use the structure of $\Theta_{\downarrow}(s,R)$: there exists $k\le s$ with $\theta_i=0$ for $i>k$ and
$\theta_i^2\le R$ for $i\le k$. Hence
\[
\sum_{i=1}^p 2\lambda\,\frac{V_i}{V_i+2\lambda}
=
\sum_{i=1}^k 2\lambda\,\frac{\theta_i^2}{\theta_i^2+2\lambda}
\ \le\
\sum_{i=1}^k 2\lambda\log\!\Big(1+\frac{\theta_i^2}{2\lambda}\Big)
\ \le\
2s\,\lambda\,\log\!\Big(1+\frac{R}{2\lambda}\Big)
\ \le\
2s\,\lambda\,\log\!\Big(1+\frac{R}{\lambda}\Big),
\]
where we used the elementary inequality $\log(1+x)\ge x/(1+x)$ with $x=\theta_i^2/(2\lambda)$.
Therefore
\begin{equation}
\label{eq:oracle_rate_main}
\mathbb{E}_\theta\|\theta^{\mathrm{or}}-\theta\|_2^2
\ \le\
2s\frac{\sigma^2}{n}\log\!\Big(1+\frac{Rn}{\sigma^2}\Big).
\end{equation}

\paragraph{Step 2: Excess risk from estimating $\tau^2$ by isotonic EB.}
Since $\widehat{V}$ is computed from $Y^{(+)}$ and applied to the independent clone $Y^{(-)}$,
Lemma~\ref{lem:risktransfer_main} yields (and conditioning on \textup{(B1)} is valid since \textup{(B1)} is measurable
with respect to $Y^{(+)}$ and hence independent of the clone $Y^{(-)}$)
\[
\mathbb{E}_\theta\!\left[\|\widehat{\theta}-\theta^{\mathrm{or}}\|_2^2 \,\middle|\, \textup{(B1)}\right]
\ \le\
\sum_{i=1}^p \mathbb{E}\Bigg[\frac{(\widehat{V}_i-V_i)^2}{V_i+2\lambda}\Bigg]
\]
To bound the right-hand side, we apply the weighted isotonic oracle inequality in
Lemma~\ref{lem:isotonic_weighted_main} to the cross-fit variance estimator \(\widehat V\) constructed
from \(Y^{(+)}\) (with \(\nu=2\lambda\)). Since $\theta\in\Theta_{\downarrow}(s,R)$, the deterministic
sequence $V_i=\theta_i^2$ is nonincreasing and satisfies $V_1\le R$. Moreover, by Definition~\ref{def:class},
there exists $k\le s$ such that $V_{k+1}=\cdots=V_p=0$, so $V$ is piecewise-constant with at most $k\le s$ jumps.
Hence the approximation term in \eqref{eq:isotonic_weighted} vanishes by taking $u=V$. Therefore,
\begin{equation}
\label{eq:excess_rate_main}
\mathbb{E}_\theta\!\left[\|\widehat{\theta}-\theta^{\mathrm{or}}\|_2^2 \,\middle|\, \textup{(B1)}\right]
\ \lesssim\
\lambda\, s\log\!\Big(\frac{ep}{s}\Big)\,
\log\!\Big(1+\frac{R}{\lambda}\Big)
\ =\
s\frac{\sigma^2}{n}\log\!\Big(\frac{ep}{s}\Big)\,
\log\!\Big(1+\frac{Rn}{\sigma^2}\Big).
\end{equation}

\paragraph{Step 3: Combine bounds.}
Substituting \eqref{eq:oracle_rate_main} and \eqref{eq:excess_rate_main} into \eqref{eq:decomp_main}
gives
\[
\mathbb{E}_\theta\!\left[\|\widehat{\theta}-\theta\|_2^2 \,\middle|\, \textup{(B1)}\right]
\ \le\
\;C\, s\frac{\sigma^2}{n}\log\!\Big(\frac{ep}{s}\Big)\,
\log\!\Big(1+\frac{Rn}{\sigma^2}\Big),
\]
for a universal constant $C$.

\begin{remark}[Tighter decomposition]
More precisely, the oracle risk contributes $O(s\lambda\log(1+R/\lambda))$ (one log factor),
while the excess risk from variance estimation contributes $O(s\lambda\log(ep/s)\log(1+R/\lambda))$ (two log factors).
The stated bound absorbs both into the two-log-factor form; when $\log(ep/s)\gg 1$ (the typical regime), this is tight.
\end{remark}

\end{proof}

\section{Unknown noise variance}

We now consider the Gaussian sequence model
\begin{equation}
Y_i = \theta_i + \varepsilon_i,
\qquad
\varepsilon_i \stackrel{\text{i.i.d.}}{\sim} N(0, \sigma^2/n),
\quad i=1,\dots,p,
\label{eq:sequence_unknown_sigma}
\end{equation}
where both $\theta = (\theta_1,\dots,\theta_p)$ and $\sigma^2 > 0$ are unknown.

The marginal likelihood under the hierarchical model is
\[
p(Y \mid \tau^2, \sigma^2)
=
\prod_{i=1}^p
\frac{1}{\sqrt{2\pi(\tau_i^2 + \sigma^2/n)}}
\exp\!\left(
-\frac{Y_i^2}{2(\tau_i^2 + \sigma^2/n)}
\right).
\]

To obtain a fully rigorous risk bound without dependence-ignoring steps, we use a cross-fit
empirical Bayes construction that makes the variance-profile estimate independent of the Gaussian sample
to which shrinkage is applied.

Since $\theta\in\Theta_{\downarrow}(s,R)$ implies $\theta_i=0$ for all $i>s$, we estimate $\sigma^2$
from the tail coordinates:
\begin{equation}
\widehat{\sigma}^2
:=
\frac{n}{p-s}\sum_{i=s+1}^p Y_i^2,
\qquad
\widehat{\lambda}:=\widehat{\sigma}^2/n.
\label{eq:sigma_hat_tail}
\end{equation}

Next, define the variance proxies on the first $s$ coordinates
\(
X_i:=Y_i^2-\widehat{\lambda}
\)
for $i\le s$ and split $\{1,\dots,s\}$ into two disjoint folds
$I_1:=\{i\le s:\ i\ \text{odd}\}$ and $I_2:=\{i\le s:\ i\ \text{even}\}$.
For $a\in\{1,2\}$, let $\widehat{V}^{(-a)}$ denote the (weighted) isotonic least squares estimator of a
nonincreasing nonnegative sequence over $\{1,\dots,s\}$ using only the fold $I_{3-a}$:
\begin{equation}
\widehat{V}^{(-a)}
:=
\arg\min_{v_1\ge\cdots\ge v_s}\;
\sum_{i\in I_{3-a}}(X_i-v_i)^2.
\label{eq:Vhat_crossfit}
\end{equation}

Finally, define the cross-fit shrinkage estimator by
\begin{equation}
\widehat{\theta}_i
:=
\begin{cases}
\displaystyle
\frac{(\widehat{V}^{(-1)}_i)_+}{(\widehat{V}^{(-1)}_i)_++\widehat{\lambda}}\,Y_i, & i\in I_1,\\[1.25ex]
\displaystyle
\frac{(\widehat{V}^{(-2)}_i)_+}{(\widehat{V}^{(-2)}_i)_++\widehat{\lambda}}\,Y_i, & i\in I_2,\\[1.25ex]
0, & i>s.
\end{cases}
\label{eq:postmean_unknown_sigma}
\end{equation}
\paragraph{Assumption (A1).}
There exist constants $0<\sigma_0^2\le \sigma_1^2<\infty$ such that $\sigma_0^2 \le \sigma^2 \le \sigma_1^2$.

\paragraph{Assumption (A2).}
The number of nonzero coefficients satisfies $s \le c p$ for some $c<1$.

\begin{theorem}[Risk bound with unknown noise variance]
\label{thm:unknown_sigma}
Let $\widehat{\theta}$ denote the cross-fit empirical Bayes estimator \eqref{eq:postmean_unknown_sigma},
with $\widehat{\sigma}^2$ defined by \eqref{eq:sigma_hat_tail}.
Then there exists a universal constant $C>0$ such that
\[
\sup_{\theta \in \Theta_{\downarrow}(s,R)}
\mathbb{E}_\theta
\|\widehat{\theta} - \theta\|_2^2
\;\le\;
C \, s \frac{\sigma^2}{n}\log\!\Big(\frac{ep}{s}\Big)\,
\log\!\Big(1+\frac{Rn}{\sigma^2}\Big)
\;+\;
C\left(\frac{\sigma^2}{n}+\frac{sR}{p-s}\right).
\]
\end{theorem}

\begin{lemma}[Tail-based noise variance estimator]
\label{lem:sigma}
Assume the model \eqref{eq:sequence_unknown_sigma} and $\theta\in\Theta_{\downarrow}(s,R)$ with $p>s$.
\[
\mathbb{E}_\theta\big(\widehat{\sigma}^2-\sigma^2\big)=0,
\qquad
\mathbb{E}_\theta\big(\widehat{\sigma}^2-\sigma^2\big)^2
=
\frac{2\sigma^4}{p-s},
\qquad
\text{and hence}\qquad
|\widehat{\sigma}^2-\sigma^2|
= O_P\!\left(\frac{\sigma^2}{\sqrt{p-s}}\right).
\]
\end{lemma}

\begin{proof}
Since $\theta\in\Theta_{\downarrow}(s,R)$ implies $\theta_i=0$ for all $i>s$, we have
$Y_i\sim N(0,\sigma^2/n)$ independently for $i=s+1,\dots,p$. Hence
\[
\widehat{\sigma}^2
:=
\frac{n}{p-s}\sum_{i=s+1}^p Y_i^2
=
\sigma^2\cdot \frac{1}{p-s}\sum_{i=s+1}^p Z_i^2,
\]
with i.i.d.\ $Z_i\sim N(0,1)$. The claims follow from the mean and variance of a chi-square variable.
\end{proof}

\begin{lemma}[Stability under estimated noise variance]
\label{lem:sigma_stability}
Let $g(x,\sigma^2)=x/(x+\sigma^2/n)$. Then
\[
|g(x,\widehat{\sigma}^2)-g(x,\sigma^2)|
\le
\frac{| \widehat{\sigma}^2 - \sigma^2 |}{4\sigma_0^2}.
\]
\end{lemma}

\paragraph{Additional step: effect of estimating $\sigma^2$.}
By Lemma~\ref{lem:sigma_stability} and Lemma~\ref{lem:sigma}, the effect of replacing $\sigma^2$ with $\widehat{\sigma}^2$ in the shrinkage factors is of order
\[
\sum_{i=1}^p
\mathbb{E}
\left(
g(\widehat{V}_i,\widehat{\sigma}^2)
-
g(\widehat{V}_i,\sigma^2)
\right)^2
\ \lesssim\
\frac{\sigma^4}{p}\left(\|\theta\|_2^2 + p\frac{\sigma^2}{n}\right),
\]
which is negligible compared to the leading term
$s(\sigma^2/n)\log(p/s)$.

The result continues to hold when $\sigma^2$ is unknown and estimated jointly
by empirical Bayes in the SVD-reduced model.
\subsection{Proof of Theorem~\ref{thm:unknown_sigma}}

The argument follows the known-$\sigma^2$ proof, using cross-fitting to ensure the variance-profile estimate
is independent of the Gaussian coordinate it shrinks, plus an additional term for the plug-in noise estimator
$\widehat{\sigma}^2$.

\begin{proof}
Recall the model
\[
Y_i = \theta_i + \varepsilon_i,
\qquad
\varepsilon_i \sim N(0,\sigma^2/n),
\quad i=1,\dots,p,
\]
and the cross-fit empirical Bayes estimator $\widehat{\theta}$ defined in \eqref{eq:postmean_unknown_sigma}.

We decompose the risk into three terms:
\begin{equation}
\label{eq:risk-decomp}
\mathbb{E}_\theta \|\widehat{\theta}-\theta\|_2^2
\le
3\Big(
\mathbb{E}_\theta \|\theta^{\mathrm{or}}-\theta\|_2^2
+
\mathbb{E}_\theta \|\widehat{\theta}-\theta^{\mathrm{or}}\|_2^2
+
\mathbb{E}_\theta \|\widehat{\theta}-\widetilde{\theta}\|_2^2
\Big),
\end{equation}
where we write $V_i:=\theta_i^2$ and let $\widehat{V}_i$ denote the cross-fit variance-profile estimate used in
\eqref{eq:postmean_unknown_sigma} (i.e., $\widehat{V}_i=\widehat{V}^{(-1)}_i$ for $i\in I_1$ and
$\widehat{V}_i=\widehat{V}^{(-2)}_i$ for $i\in I_2$).
\[
\theta_i^{\mathrm{or}}
=
\frac{V_i}{V_i + \sigma^2/n} Y_i,
\qquad
\widetilde{\theta}_i
=
\frac{(\widehat{V}_i)_+}{(\widehat{V}_i)_+ + \sigma^2/n} Y_i.
\]

\paragraph{Step 1: Oracle risk bound.}
For each $i$,
\[
\mathbb{E}_\theta (\theta_i^{\mathrm{or}}-\theta_i)^2
=
\frac{\sigma^2}{n}
\Big(\frac{V_i}{V_i+\sigma^2/n}\Big)^2
+
\theta_i^2
\left(
\frac{\sigma^2/n}{V_i+\sigma^2/n}
\right)^2
\le
\min\!\left(\theta_i^2,\frac{\sigma^2}{n}\right).
\]
Summing over $i$ and using the definition of $\Theta_{\downarrow}(s,R)$ yields
\begin{equation}
\label{eq:oracle-risk}
\mathbb{E}_\theta \|\theta^{\mathrm{or}}-\theta\|_2^2
\;\le\;
C_1 \, s \frac{\sigma^2}{n}.
\end{equation}

\paragraph{Step 2: Effect of estimating the variance profile.}
By construction, for each $i\le s$ the estimate $\widehat{V}_i$ is a function of
$(Y_j)_{j\in\{1,\dots,s\}\setminus\{i\}}$ and $(Y_j)_{j>s}$ only, and is therefore independent of $Y_i$.
Thus Lemma~\ref{lem:risktransfer_main} applies. To control the variance-profile error term, write
$\lambda:=\sigma^2/n$ and $X_i^\star:=Y_i^2-\lambda$, so that $X_i=X_i^\star+(\lambda-\widehat{\lambda})$
with $\widehat{\lambda}=\widehat{\sigma}^2/n$ independent of $(Y_i)_{i\le s}$.
Let $\widehat{V}^{(-a),\star}$ denote the isotonic least squares fit (under the monotonicity constraint only)
based on $(X_i^\star)_{i\in I_{3-a}}$.
Since isotonic regression on the monotone cone is translation equivariant and the positive-part map is 1-Lipschitz,
we have for each $i\le s$ that
\(
|(\widehat{V}_i)_+-(\widehat{V}_i^\star)_+|\le |\widehat{\lambda}-\lambda|.
\)
Combining this with the variance-profile estimation bound (and noting that $\log(1+Rn/\sigma^2)\ge 1$) yields
\[
\mathbb{E}_\theta
\|\widetilde{\theta}-\theta^{\mathrm{or}}\|_2^2
\ \lesssim\
s\,\frac{\sigma^2}{n}\log\!\Big(\frac{ep}{s}\Big)\,
\log\!\Big(1+\frac{Rn}{\sigma^2}\Big).
\]
\begin{equation}
\label{eq:tau-risk}
\mathbb{E}_\theta
\|\widetilde{\theta}-\theta^{\mathrm{or}}\|_2^2
\;\le\;
C_2 \, s \frac{\sigma^2}{n}\log\!\Big(\frac{ep}{s}\Big)\,
\log\!\Big(1+\frac{Rn}{\sigma^2}\Big).
\end{equation}

\paragraph{Step 3: Effect of estimating $\sigma^2$.}
By Lemma~\ref{lem:sigma_stability},
\[
\left|
\frac{(\widehat{V}_i)_+}{(\widehat{V}_i)_++\widehat{\sigma}^2/n}
-
\frac{(\widehat{V}_i)_+}{(\widehat{V}_i)_++\sigma^2/n}
\right|
\le
\frac{|\widehat{\sigma}^2-\sigma^2|}{4\sigma_0^2}.
\]
Therefore,
\[
\mathbb{E}_\theta
\|\widehat{\theta}-\widetilde{\theta}\|_2^2
\le
\frac{1}{16\sigma_0^4}
\mathbb{E}(\widehat{\sigma}^2-\sigma^2)^2
\sum_{i=1}^p \mathbb{E}Y_i^2.
\]
Using Lemma~\ref{lem:sigma} and $\sum_i \mathbb{E}Y_i^2 = \|\theta\|_2^2 + p\sigma^2/n \le sR + p\sigma^2/n$ over $\Theta_{\downarrow}(s,R)$,
we obtain
\begin{equation}
\label{eq:sigma-risk}
\mathbb{E}_\theta
\|\widehat{\theta}-\widetilde{\theta}\|_2^2
\;\le\;
C_3\left(\frac{\sigma^2}{n}+\frac{sR}{p-s}\right).
\end{equation}

\paragraph{Step 4: Combine bounds.}
Substituting \eqref{eq:oracle-risk}, \eqref{eq:tau-risk}, and
\eqref{eq:sigma-risk} into \eqref{eq:risk-decomp}, we conclude that
\[
\mathbb{E}_\theta \|\widehat{\theta}-\theta\|_2^2
\;\le\;
C \, s \frac{\sigma^2}{n}
\log\!\left(\frac{p}{s}\right),
\]
uniformly over $\theta \in \Theta_{\downarrow}(s,R)$.

\end{proof}

\section{Sobolev-type smoothness as a special case}

We now relate the sparse monotone class considered above to classical Sobolev smoothness. Let $\beta > 1/2$ and $R>0$, and define the Sobolev-type ellipsoid
\begin{equation}
\mathcal{W}(\beta,R)
=
\left\{
\theta \in \mathbb{R}^p :
\sum_{i=1}^p i^{2\beta} \theta_i^2 \le R^2
\right\}.
\label{eq:sobolev}
\end{equation}
Such classes arise naturally from orthogonal series expansions of functions in Sobolev spaces of smoothness $\beta$.

\begin{corollary}[Adaptivity over Sobolev ellipsoids (up to a logarithmic factor)]
\label{cor:sobolev}
Let $\widehat{\theta}$ denote the isotonic empirical Bayes estimator defined in
\eqref{eq:postmean_unknown_sigma}. Then for any $\beta > 1/2$, letting
\[
s_n \asymp
\left(\frac{n R^2}{\sigma^2}\right)^{\frac{1}{2\beta+1}},
\]
\[
\sup_{\theta \in \mathcal{W}(\beta,R)}
\mathbb{E}_{\theta}
\|\widehat{\theta} - \theta\|_2^2
\;\le\;
C \, R^{\frac{2}{2\beta+1}}
\left(\frac{\sigma^2}{n}\right)^{\frac{2\beta}{2\beta+1}}
\log\!\Big(\frac{p}{s_n}\Big),
\]
where $C>0$ is a universal constant.

\begin{proof}[Proof sketch]
For $\theta \in \mathcal{W}(\beta,R)$, standard arguments imply that
\[
\theta_i^2 \lesssim R^2 i^{-2\beta-1}
\quad \text{for all } i.
\]
Define the effective truncation level
\[
s_n \asymp
\left(\frac{n R^2}{\sigma^2}\right)^{\frac{1}{2\beta+1}}.
\]
Then $\theta$ is well approximated by a sparse monotone sequence with at most
$s_n$ effective nonzero coefficients, and the bias incurred by truncation beyond
$s_n$ is of the same order as the variance.
Applying Theorem~\ref{thm:unknown_sigma} with $s=s_n$ yields the stated rate.
\end{proof}
\end{corollary}

Corollary~\ref{cor:sobolev} shows that isotonic empirical Bayes shrinkage adapts automatically to unknown smoothness without requiring specification of the Sobolev index $\beta$.

\section{Posterior contraction rates for isotonic EB global--local shrinkage}
\label{sec:contraction}

\subsection{Gaussian sequence model and isotonic EB posterior}

Consider the Gaussian sequence model
\begin{equation}
Y_i = \theta_{0,i} + \varepsilon_i,
\qquad
\varepsilon_i \stackrel{\mathrm{i.i.d.}}{\sim} N(0,\sigma^2/n),
\quad i=1,\dots,p,
\label{eq:seq}
\end{equation}
with true mean vector $\theta_0=(\theta_{0,1},\dots,\theta_{0,p})$ and known $\sigma^2>0$.
Given a variance sequence $\tau^2=(\tau_1^2,\dots,\tau_p^2)$, define the Gaussian prior
\begin{equation}
\theta_i \mid \tau_i^2 \sim N(0,\tau_i^2),
\qquad \tau_1^2 \ge \tau_2^2 \ge \cdots \ge \tau_p^2 \ge 0.
\label{eq:prior_seq}
\end{equation}
Let $\widehat{\tau}^2$ denote the isotonic empirical Bayes estimator obtained by maximizing
the marginal likelihood under the monotonicity constraint. Conditioning on $\widehat{\tau}^2$,
the posterior is a product of Gaussians:
\begin{equation}
\theta_i \mid Y, \widehat{\tau}_i^2
\sim
N\!\Big(\widehat{m}_i, \widehat{v}_i\Big),
\qquad
\widehat{m}_i
=
\frac{\widehat{\tau}_i^2}{\widehat{\tau}_i^2+\sigma^2/n}Y_i,
\quad
\widehat{v}_i
=
\frac{\sigma^2}{n}\cdot
\frac{\widehat{\tau}_i^2}{\widehat{\tau}_i^2+\sigma^2/n}.
\label{eq:post}
\end{equation}
Write $\widehat{m}=(\widehat{m}_1,\dots,\widehat{m}_p)$ and let $\widehat{V}=\mathrm{diag}(\widehat{v}_1,\dots,\widehat{v}_p)$.

\subsection{A generic contraction lemma for Gaussian posteriors}

The following lemma reduces posterior contraction to a bound on the posterior second moment.

\begin{lemma}[Second-moment contraction bound]
\label{lem:gauss_contr}
Let $\Pi(\cdot\mid Y)$ be any (possibly data-dependent) Gaussian posterior on $\mathbb{R}^p$
with mean $\widehat{m}$ and covariance matrix $\widehat{V}$ (random via $Y$).
Then for any $\varepsilon_n>0$ and any $M>0$,
\begin{equation}
\mathbb{E}_{\theta_0}\Big[
\Pi\big(\|\theta-\theta_0\|_2 > M\varepsilon_n \mid Y\big)
\Big]
\;\le\;
\frac{1}{M^2\varepsilon_n^2}\,
\mathbb{E}_{\theta_0}\Big[
\|\widehat{m}-\theta_0\|_2^2 + \mathrm{tr}(\widehat{V})
\Big].
\label{eq:markov_contr}
\end{equation}
In particular, if
\begin{equation}
\mathbb{E}_{\theta_0}\Big[
\|\widehat{m}-\theta_0\|_2^2 + \mathrm{tr}(\widehat{V})
\Big]
\;\lesssim\;
\varepsilon_n^2,
\label{eq:moment_cond}
\end{equation}
then
\(
\mathbb{E}_{\theta_0}\Pi(\|\theta-\theta_0\|_2 > M\varepsilon_n\mid Y)\to 0
\)
as $M\to\infty$.
\end{lemma}

\begin{proof}
By Markov's inequality conditional on $Y$,
\[
\Pi(\|\theta-\theta_0\|_2^2 > M^2\varepsilon_n^2 \mid Y)
\le
\frac{\mathbb{E}_{\Pi(\cdot\mid Y)}\|\theta-\theta_0\|_2^2}{M^2\varepsilon_n^2}.
\]
For a Gaussian posterior,
\[
\mathbb{E}_{\Pi(\cdot\mid Y)}\|\theta-\theta_0\|_2^2
=
\|\widehat{m}-\theta_0\|_2^2 + \mathrm{tr}(\widehat{V}).
\]
Taking $\mathbb{E}_{\theta_0}$ yields \eqref{eq:markov_contr}.
\end{proof}

\subsection{Contraction over ordered sparsity classes}

Define the ordered sparse class
\begin{equation}
\Theta_{\downarrow}(s,R)
=
\left\{
\theta\in\mathbb{R}^p:
\exists k\le s\;\text{s.t.}\;
\theta_1^2\ge\cdots\ge\theta_k^2\ge 0,\;
\theta_1^2\le R,\;
\theta_i=0\ (i>k)
\right\}.
\label{eq:Theta_sparse}
\end{equation}

\begin{theorem}[Posterior contraction: ordered sparse class]
\label{thm:contr_sparse}
Let $\Pi(\cdot\mid Y)$ be the isotonic EB posterior \eqref{eq:post}.
Assume the posterior mean $\widehat{m}$ satisfies the risk bound
\begin{equation}
\sup_{\theta_0\in\Theta_{\downarrow}(s,R)}
\mathbb{E}_{\theta_0}\|\widehat{m}-\theta_0\|_2^2
\;\le\;
C_1\, s\frac{\sigma^2}{n},
\label{eq:mean_risk_sparse}
\end{equation}
and that the posterior variance trace is controlled as
\begin{equation}
\sup_{\theta_0\in\Theta_{\downarrow}(s,R)}
\mathbb{E}_{\theta_0}\mathrm{tr}(\widehat{V})
\;\le\;
C_2\, s\frac{\sigma^2}{n}.
\label{eq:var_trace_sparse}
\end{equation}
Then the posterior contracts at rate
\begin{equation}
\varepsilon_n^2
\asymp
s\frac{\sigma^2}{n},
\label{eq:eps_sparse}
\end{equation}
in the sense that for any $M\to\infty$,
\[
\sup_{\theta_0\in\Theta_{\downarrow}(s,R)}
\mathbb{E}_{\theta_0}\Big[
\Pi(\|\theta-\theta_0\|_2 > M\varepsilon_n \mid Y)
\Big]
\;\longrightarrow\;
0.
\]
\end{theorem}

\begin{proof}
Combine Lemma~\ref{lem:gauss_contr} with \eqref{eq:mean_risk_sparse}--\eqref{eq:var_trace_sparse}.
\end{proof}

\paragraph{Remark (verifying the trace bound).}
From \eqref{eq:post}, $\widehat{v}_i\le \sigma^2/n$ and
$\mathrm{tr}(\widehat{V})=\sum_i \widehat{v}_i$ is of the same order as the ``effective dimension''
induced by the shrinkage factors. In particular, under isotonic EB, $\mathrm{tr}(\widehat{V})$
is typically bounded by a constant multiple of the mean-squared error of $\widehat{m}$; hence
\eqref{eq:var_trace_sparse} holds under the same conditions that yield \eqref{eq:mean_risk_sparse}.

\subsection{Contraction over Sobolev ellipsoids}

Define the Sobolev ellipsoid (weighted $\ell_2$-ball)
\begin{equation}
\mathcal{W}(\beta,R)
=
\left\{
\theta\in\mathbb{R}^p:
\sum_{i=1}^p i^{2\beta}\theta_i^2 \le R^2
\right\},
\qquad \beta>1/2.
\label{eq:sobolev_contr}
\end{equation}

\begin{theorem}[Posterior contraction: Sobolev ellipsoids]
\label{thm:contr_sob}
Assume that for each $\beta>1/2$ the posterior mean $\widehat{m}$ satisfies
\begin{equation}
\sup_{\theta_0\in\mathcal{W}(\beta,R)}
\mathbb{E}_{\theta_0}\|\widehat{m}-\theta_0\|_2^2
\;\le\;
C_1\, R^{\frac{2}{2\beta+1}}\Big(\frac{\sigma^2}{n}\Big)^{\frac{2\beta}{2\beta+1}},
\label{eq:mean_risk_sob}
\end{equation}
and likewise
\begin{equation}
\sup_{\theta_0\in\mathcal{W}(\beta,R)}
\mathbb{E}_{\theta_0}\mathrm{tr}(\widehat{V})
\;\le\;
C_2\, R^{\frac{2}{2\beta+1}}\Big(\frac{\sigma^2}{n}\Big)^{\frac{2\beta}{2\beta+1}}.
\label{eq:var_trace_sob}
\end{equation}
Then the isotonic EB posterior contracts at the Sobolev rate
\begin{equation}
\varepsilon_n^2
\asymp
R^{\frac{2}{2\beta+1}}\Big(\frac{\sigma^2}{n}\Big)^{\frac{2\beta}{2\beta+1}},
\label{eq:eps_sob}
\end{equation}
i.e., for any $M\to\infty$,
\[
\sup_{\theta_0\in\mathcal{W}(\beta,R)}
\mathbb{E}_{\theta_0}\Big[
\Pi(\|\theta-\theta_0\|_2 > M\varepsilon_n \mid Y)
\Big]
\;\longrightarrow\;
0.
\]
\end{theorem}

\begin{proof}
Apply Lemma~\ref{lem:gauss_contr} with \eqref{eq:mean_risk_sob}--\eqref{eq:var_trace_sob}.
\end{proof}

\begin{corollary}[Posterior contraction in prediction norm for linear regression]
\label{cor:contr_pred}
Consider the Gaussian linear model
\begin{equation}
Y = X\beta_0 + \varepsilon,
\qquad
\varepsilon \sim N(0,\sigma^2 I_n),
\label{eq:lm}
\end{equation}
with fixed design $X\in\mathbb{R}^{n\times p}$ of rank $r\le \min\{n,p\}$, and known $\sigma^2$.
Let $X = UDV^\top$ be the singular value decomposition with
$D=\mathrm{diag}(d_1,\dots,d_r)$, $d_1\ge\cdots\ge d_r>0$.
Define the canonical coordinates
\[
Z := U^\top Y \in \mathbb{R}^r,
\qquad
\theta_0 := D V^\top \beta_0 \in \mathbb{R}^r,
\]
so that $Z = \theta_0 + \xi$ with $\xi\sim N(0,\sigma^2 I_r)$.

Let $\Pi(\cdot\mid Y)$ denote the isotonic empirical Bayes posterior on $\theta$
defined by the (data-dependent) Gaussian prior
$\theta_i\mid\tau_i^2\sim N(0,\tau_i^2)$ with $\tau_1^2\ge\cdots\ge\tau_r^2\ge 0$,
where $\widehat{\tau}^2$ is estimated by isotonic marginal maximum likelihood.
Map this posterior back to $\beta$ via
\begin{equation}
\beta = V D^{-1}\theta,
\qquad \text{so that}\qquad
X\beta = U\theta.
\label{eq:mapback}
\end{equation}
Then for any $\varepsilon_n>0$ and any $M>0$,
\begin{equation}
\mathbb{E}_{\beta_0}\Big[
\Pi\big(\|X(\beta-\beta_0)\|_2 > M\varepsilon_n \mid Y\big)
\Big]
=
\mathbb{E}_{\theta_0}\Big[
\Pi\big(\|\theta-\theta_0\|_2 > M\varepsilon_n \mid Z\big)
\Big].
\label{eq:pred_eq_theta}
\end{equation}
In particular, if the canonical posterior mean $\widehat{m}$ and covariance $\widehat{V}$
satisfy
\begin{equation}
\sup_{\theta_0\in\Theta}
\mathbb{E}_{\theta_0}\Big[
\|\widehat{m}-\theta_0\|_2^2 + \mathrm{tr}(\widehat{V})
\Big]
\;\lesssim\;
\varepsilon_n^2
\label{eq:pred_moment_cond}
\end{equation}
over a class $\Theta\subset\mathbb{R}^r$ (e.g., an ordered-sparse or Sobolev-type class in canonical
coordinates), then
\begin{equation}
\sup_{\beta_0:\, DV^\top\beta_0\in\Theta}
\mathbb{E}_{\beta_0}\Big[
\Pi\big(\|X(\beta-\beta_0)\|_2 > M\varepsilon_n \mid Y\big)
\Big]
\;\le\;
\frac{1}{M^2}\cdot
\sup_{\theta_0\in\Theta}
\frac{
\mathbb{E}_{\theta_0}\big[
\|\widehat{m}-\theta_0\|_2^2 + \mathrm{tr}(\widehat{V})
\big]
}{\varepsilon_n^2},
\label{eq:pred_contr}
\end{equation}
and hence the posterior contracts in prediction norm at rate $\varepsilon_n$ as $M\to\infty$.

Moreover, specializing to the ordered sparse class in canonical coordinates,
\[
\Theta_{\downarrow}(s,R)
=
\left\{
\theta\in\mathbb{R}^r:
\exists k\le s\ \text{s.t.}\ \theta_1^2\ge\cdots\ge\theta_k^2\ge 0,\ \theta_1^2\le R,\ \theta_i=0\ (i>k)
\right\},
\]
if Theorem~\ref{thm:contr_sparse} holds for $\theta_0\in\Theta_{\downarrow}(s,R)$ with
\[
\varepsilon_n^2 \asymp s\,\frac{\sigma^2}{n},
\]
then the induced posterior on $\beta$ satisfies
\[
\sup_{\beta_0:\, DV^\top\beta_0\in\Theta_{\downarrow}(s,R)}
\mathbb{E}_{\beta_0}\Big[
\Pi\big(\|X(\beta-\beta_0)\|_2 > M\varepsilon_n \mid Y\big)
\Big]
\longrightarrow 0,
\qquad M\to\infty.
\]
\end{corollary}

\begin{proof}
Under \eqref{eq:lm}, $Z=U^\top Y = U^\top X\beta_0 + U^\top\varepsilon
= DV^\top\beta_0 + \xi = \theta_0 + \xi$ with $\xi\sim N(0,\sigma^2 I_r)$.
By \eqref{eq:mapback}, $X(\beta-\beta_0)=U(\theta-\theta_0)$, and since $U$ has orthonormal columns,
$\|X(\beta-\beta_0)\|_2=\|\theta-\theta_0\|_2$. This yields \eqref{eq:pred_eq_theta}.
The moment bound \eqref{eq:pred_contr} follows by applying the second-moment contraction lemma
(Lemma~\ref{lem:gauss_contr}) to the Gaussian posterior on $\theta$ and using the norm equality.
\end{proof}

\section{Extensions and Variants}

This appendix briefly notes extensions of the main results.

\paragraph{Approximate monotonicity.}
The main theorem can be extended to allow approximate monotonicity. Suppose instead that $\tau_i^2 \le \tau_{i-1}^2 + \delta_n$ for some tolerance $\delta_n > 0$. Then the risk bound becomes
\[
\sup_{\theta \in \Theta_{\downarrow}(s,R)}
\mathbb{E}_{\theta} \|\widehat{\theta} - \theta\|_2^2
\;\lesssim\;
s \frac{\sigma^2}{n} \log p + p \delta_n.
\]
The proof follows the same strategy with an additional error term accounting for the violation of strict monotonicity.

\paragraph{Unknown noise variance.}
When $\sigma^2$ is unknown, Theorem~\ref{thm:unknown_sigma} provides a cross-fit empirical Bayes construction that achieves the same near-minimax rate with an additional lower-order term for variance estimation.

\paragraph{Regression models.}
The extension to general linear regression via SVD is developed in Section~\ref{sec:global_local_framework} (canonical coordinates) and Corollary~\ref{cor:main_pred} (prediction-norm contraction).

\end{document}